\documentclass[preprint]{elsarticle}

\usepackage[hidelinks]{hyperref}

\journal{Expert Systems with Applications}

\usepackage{color}
\usepackage{amsmath}
\usepackage{amsthm}
\usepackage{amssymb}
\usepackage{array}
\usepackage{algorithm}
\usepackage{algpseudocode}
\usepackage[caption=false,font=footnotesize]{subfig}
\usepackage{graphicx}

\newtheorem{theorem}{Theorem}

\newtheorem{definition}{Definition}
\newtheorem{example}{Example}
\DeclareMathOperator*{\argmax}{arg\,max}
\DeclareMathOperator*{\argmin}{arg\,min}

\begin{document}

\begin{frontmatter}

\title{Differentially Private Random Decision Forests\\ using Smooth Sensitivity}

\author[mymainaddress]{Sam Fletcher\corref{mycorrespondingauthor}}
\cortext[mycorrespondingauthor]{Corresponding author}
\ead{sam.pt.fletcher@gmail.com}

\author[mymainaddress]{Md Zahidul Islam}
\ead{zislam@csu.edu.au}

\address[mymainaddress]{School of Computing and Mathematics, Charles Sturt University, Bathurst, Australia}

\begin{abstract}
We propose a new differentially-private decision forest algorithm that minimizes both the number of queries required, and the sensitivity of those queries. To do so, we build an ensemble of random decision trees that avoids querying the private data except to find the majority class label in the leaf nodes. Rather than using a count query to return the class counts like the current state-of-the-art, we use the Exponential Mechanism to only output the class label itself. This drastically reduces the sensitivity of the query -- often by several orders of magnitude -- which in turn reduces the amount of noise that must be added to preserve privacy. Our improved sensitivity is achieved by using ``smooth sensitivity'', which takes into account the specific data used in the query rather than assuming the worst-case scenario. We also extend work done on the optimal depth of random decision trees to handle continuous features, not just discrete features. This, along with several other improvements, allows us to create a differentially private decision forest with substantially higher predictive power than the current state-of-the-art.
\end{abstract}


\end{frontmatter}

\section{Introduction}

Information about people is becoming increasingly valuable in the data-driven society of the 21st century. The ability to extract knowledge from data allows government and industry bodies to make informed decisions. This can be anything from monitoring traffic data and predicting congestion or intelligently controlling traffic lights, to reading large amounts of medical data and finding new disease patterns or predicting patient re-admission probabilities, to predicting financial market fluctuations. Many fields of science intersect when mining data; machine learning, statistics, and database systems are intertwined in order to produce useful, usable information \citep{Breiman2001a}. ``Data mining algorithms'' covers a wide range of possible algorithms; some produce a set of humanly-readable patterns discovered in the data, some produce a classification or regression model that enable predictions to be made about the future, and others detect anomalies in the data. In this paper, we focus on a system that takes large amounts of data as input, and outputs a classification model. This model then in turn takes new data as input, and outputs predictions (i.e., classifications) about how that data should be labeled. Fig.~\ref{fig:DM} presents a high-level view of this system.

Unfortunately, sometimes data models are not so easily built and knowledge is not so easily discovered. Sometimes, the privacy of the people being data mined needs to be taken into account. Whether it is the government collecting information about its citizens, a business collecting information about its present and future customers, or a health organization collecting information about illnesses or treatments, the privacy of the individuals being analyzed is a human right \citep{UNGeneralAssembly1948} that needs protecting. How exactly we go about protecting the privacy of individuals while also building models and discovering knowledge is a large question, and one that this paper weighs in on. We propose a data mining algorithm capable of classifying previously unseen data with high accuracy, built in a way that protects the privacy of every single person recorded in the training data.

Sometimes data is available to the public, such as through a government project like a census. Other data is collected and owned by businesses, where employees have access to it. In both cases, data subjects feel more comfortable if they know that personally-identifying information cannot be accessed by anyone without their explicit permission. In some cases the government may mandate privacy protections; in other cases, businesses may want to offer privacy protections to gain financial advantages, by encouraging more data subjects to provide their information. Differential privacy \citep{Dwork2006,Dwork2014} offers a way for both governments and businesses to guarantee privacy protections to every individual person who has information in a dataset. It is through the framework of differential privacy that our proposals are designed, taking advantage of the privacy guarantees it can make to each individual person in a dataset.

\begin{figure}[t]
	\centering
	\includegraphics[width=4.4in]{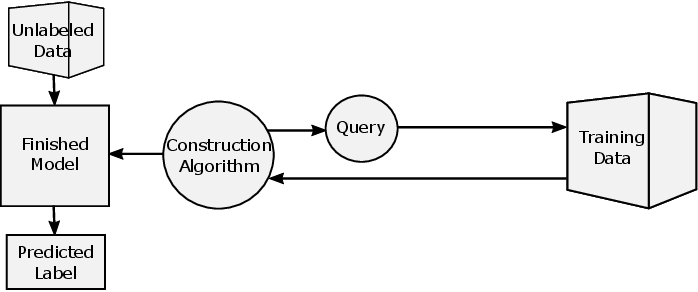}
	\caption{A high-level diagram of a data mining algorithm building a classification model from data.}
	\label{fig:DM}
\end{figure}

Fig.~\ref{fig:pipeline} presents a high-level view of the pipeline used by differential privacy; an algorithm (henceforth generalized to ``the user'') submits a query to the dataset, the dataset calculates the answer to the query, and then a differential privacy mechanism modifies the answer in a way that makes it mathematically impossible to detect if any specific individual is in the dataset or not. The differentially-private output is then returned to the user, and can be used in whatever calculations the user wishes at no further privacy cost. In this system (visualized in Fig.~\ref{fig:pipeline}), differential privacy does not need to be considered when using the final model; its job is done. Different implementations of differential privacy may add more or less noise to the answer, outputting answers that are closer or further from the true answer. We propose a new implementation of differential privacy that is designed for a specific data mining scenario: decision trees \citep{Breiman1984,Han2006a}. By taking advantage of what makes decision trees different from generic queries, we are able to produce an ensemble of decision trees (i.e., a decision forest) that can classify unseen data with substantially higher accuracy than the current state-of-the-art in differentially private decision trees and forests \citep{Friedman2010,Jagannathan2012,Jagannathan2013,Fletcher2015b,Fletcher2015c,Mohammed2015,Rana2016}.

Out of the possible classification algorithms \citep{Han2006a} that could be used to learn from data, decision trees are a good choice for minimizing the drawbacks of providing differential privacy. The main factors that dictate the impact of differential privacy on a classifier are: (a) the size of the privacy budget; (b) the number of queries that are required to build the classifier; and (c) the sensitivity of those queries to small changes in the data. (a) The size of the budget is outside the control of the user; all a classifier can aim to do is produce high-quality predictions with as small a budget as possible. Given a reasonably large dataset, we achieve high accuracy with a budget as small as $\epsilon=0.1$. (b) The number of queries dictates how much the budget needs to be divided up. Unless the queries are applied to disjoint subsets of the data, they compose and each costs a portion of the budget. This is something we can control, and our algorithm is built using the absolute minimum: one query, repeated across many disjoint subsets of the data, each able to use the entire privacy budget. (c) Finally, the impact of differential privacy is dictated by the sensitivity of the queries. Our work is the first to apply the concept of smooth sensitivity to decision trees. We demonstrate that our single query is often orders of magnitude less sensitive than the queries used by other differentially private decision trees, and equally sensitive in the worst-case scenario.

Additionally, the observations we make and the theorems we prove are not limited to decision trees, but can be used in any differentially-private scenario that aims to output discrete answers to queries about the most or least frequent item in a set.

\begin{figure}[t]
	\centering
	\includegraphics[width=4.4in]{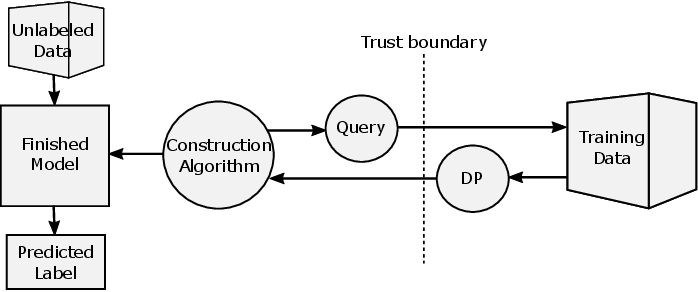}
	\caption{A high-level diagram of a data mining algorithm interfacing with private data, using differential privacy (DP). An untrusted user submits queries to a private data source and receives differentially-private answers.}
	\label{fig:pipeline}
\end{figure}

\subsection{Problem Statement} \label{subsec:problem}

A data owner possesses a two-dimensional dataset $x$ of $n$ records (i.e., tuples, rows) and $m$ features (i.e., attributes, columns). We assume that the feature schema is publicly available, including the domain of each of the features. The data owner provides an untrusted user with limited access to the data, giving the user a specific ``privacy budget'' $\epsilon$ to work with. Each time the user queries the dataset with some function $f(x)$, an amount of the privacy budget is spent depending on how invasive the query was to the privacy of the individuals described in the dataset. Once all the privacy budget is spent, the user loses access to the data forever. The user may have any number of specific goals, but in general they are trying to discover as much knowledge as they can with the budget they are given. Note that other research has explored a different scenario where the data owner wishes to publish a differentially-private version of their dataset to the public \citep{Blum2013,Li2014}, with some using decision trees to do so \citep{Mohammed2011}. This scenario is different from the one we address, where a data mining algorithm accesses the private data directly.

We use strong differential privacy, with no additional assumptions about the data. Nor do we use a non-zero $\delta$ when considering the more general $(\epsilon,\delta)$-differential privacy \citep{Dwork2014}, or otherwise weaken the definition \citep{Rana2016}.

\subsection{Our Contribution} \label{subsec:contribution}

In this paper, we propose a differentially-private decision forest algorithm that makes very efficient use of the privacy budget $\epsilon$ to output a classifier with high prediction accuracy. We achieve this by proposing a query in Section~\ref{subsec:majority} that outputs the most frequent label in some subset $x_i$ of the data with high probability, and using this query in each leaf node of all the trees in a forest. We prove that this query has low sensitivity, making it reliable even without a large privacy budget. This proof is generalized to the non-binary case, and tested on several datasets that have more than two class labels. It is also generalized to any scenario where a differentially private query aims to output the most (or least) frequent item in a set.

We also extend the work done by \citet{Fan2003}, where combinatorial reasoning is used to calculate the optimal depth of the decision trees. This work was only applicable when all the features were discrete; we extend it by proving the optimal tree depth needed for continuous features in Section~\ref{subsec:depth}.

The reliability and accuracy of our differentially-private decision forest is further improved using theoretical and empirical observations in Sections~\ref{subsec:disjoint} and \ref{subsec:numtrees}. We provide empirical results throughout the paper, demonstrating the real-world effect of our theory. The methodology of our experiments is laid out in Section~\ref{subsec:methodology}.

We finish by demonstrating that our algorithm achieves significantly better prediction accuracy than other similar algorithms, both statistically and practically, in Section~\ref{sec:experiments}. We show that even when competing against a decision forest that uses a weaker definition of differential privacy \citep{Rana2016}, our algorithm performs very well. This disputes any arguments that traditional $(\epsilon,0)$-differential privacy is too strict to produce high quality results. We believe our work is the first decision tree algorithm to simultaneously provide rigorous privacy guarantees while also outputting a model with useable utility in real-world applications. Other models achieve good results if the privacy budget is high enough, but ours is the first to perform well with a realistic budget.

We contextualize our performance on various datasets with the results obtained by \citet{Breiman2001}'s Random Forest algorithm. This algorithm is very much \emph{non-private}; that is, it makes no attempt whatsoever to protect the privacy of the people in the dataset. This of course means that it can produce more accurate models than a \emph{private} model can ever achieve, and the aim of our work (and others like it) is instead to accomplish the much more difficult task of balancing two fundamentally adversarial concepts -- knowledge discovery and individual privacy -- to produce a high quality model. We discuss the ramifications of our findings in Section~\ref{sec:discussion}. A full implementation of our algorithm is available online at \emph{http://samfletcher.work/code/} and \emph{http://csusap.csu.edu.au/{\raise.17ex\hbox{$\scriptstyle\mathtt{\sim}$}}zislam/}.

\section{Preliminaries} \label{sec:prelim}

In this section we briefly provide the necessary background knowledge for our work. We define and explain differential privacy and some of its properties and mechanisms in Sections \ref{subsec:DP} and \ref{subsec:smooth}. We describe how decision trees are built and used in Section~\ref{subsec:RDTs}. We then lay out the methodology used in all of our experiments, which are presented throughout the paper to accompany the theoretical work.

\subsection{Differential Privacy} \label{subsec:DP}

Differential privacy is a definition that makes a promise to each individual who has personal data in a dataset $x$: ``You will not be affected, adversely or otherwise, by allowing your data to be used in any analysis of the data, no matter what other analyses, datasets, or information sources are available'' \citep{Dwork2014}. The output of any query performed on a dataset $x$ will be very similar to the output of the same query performed on dataset $y$, where $x$ and $y$ only differ by any one individual. Note that this does not promise that an attacker will not learn anything about an individual; only that any information they do learn could have been learned even if the individual was not in the dataset. More formal definitions, including some additional properties of differential privacy, are presented below. Fig.~\ref{fig:pipeline} provides a high-level view of how a user interfaces with a dataset $x$ using differential privacy. We refer the reader to \citet{Dwork2014} for a more thorough introduction of differential privacy.

We write the following definitions in terms of some query (i.e., function) $f$ submitted to a dataset $x$ describing $n$ records (i.e., individuals) from a universe $D$. We define the distance between two datasets $x$ and $y$ using the Hamming distance $||x-y||_1$, which equals the number of records that would have to be added or removed from $x$ before it is identical to $y$. If the Hamming distance $||x-y||_1$ equals 1, the datasets differ by only one record and we call $x$ and $y$ ``neighbors''. We denote the cardinality (i.e., size, or number of elements) of $x$ as $|x|$.

\begin{definition}[Differential Privacy \citep{Dwork2006}] \label{def:DP}
A query $f$ is $\epsilon$-differentially private if for all outputs $g\subseteq Range(f)$ and for all data $x,y\in D^n$ such that $||x-y||_1\leq 1$: 
\begin{equation}
Pr(f(x)=g) \leq e^\epsilon \times Pr(f(y)=g) \enspace .
\end{equation}
\end{definition}

The parameter $\epsilon$ can be considered as a ``cost'', with multiple costs summing together:
\begin{definition}[Composition \citep{McSherry2007}] \label{def:composition}
The application of all queries $\{f_i(x)\}$, each satisfying $\epsilon_i$-differential privacy, satisfies $\sum\limits_i {\epsilon_i}$-differential privacy.
\end{definition}
If the same query is submitted to multiple subsets of the data, with no overlapping records, the costs do not need to be summed:
\begin{definition}[Parallel Composition \citep{McSherry2009}] \label{def:parallel}
For disjoint subsets $x_i\subset x$, let query $f(x_i)$ satisfy $\epsilon$-differential privacy; then applying all queries $\{f(x_i)\}$ still satisfies $\epsilon$-differential privacy.
\end{definition}
In other words, the privacy cost of multiple queries applied to the same data composes (i.e., is summed together), and a single query applied to different subsets of data (with no overlapping records) can be asked in parallel at no extra cost.

There are multiple ways to spend $\epsilon$. One of the more common implementations of differential privacy is the Exponential Mechanism, which returns the ``best'' discrete output with high probability:
\begin{definition}[Exponential Mechanism \citep{McSherry2007}]\label{def:expo}
Using a scoring function $u(z,x):u\rightarrow\mathbb{R}$ where $u$ has a higher value for more preferable outputs $z\in Z$, a query $f$ satisfies $\epsilon$-differential privacy if it outputs $z$ with probability proportional to $\exp{(\frac{\epsilon u(z,x)}{2\Delta(u)})}$. That is,
\begin{equation}\label{eq:expo}
Pr(f(x)=z) \propto \exp{\left(\frac{\epsilon\times u(z,x)}{2\Delta(u)}\right)} \enspace .
\end{equation}
\end{definition}

In the above definition, $\Delta(u)$ refers to the sensitivity of $u$ \citep{Dwork2006a}. The sensitivity of a function is the maximum amount the output of the function can differ by when considering two neighboring datasets $x$ and $y$. The most simple form of sensitivity $\Delta(u)$ is global sensitivity, defined as
\begin{definition}[Global Sensitivity \citep{Dwork2006a}] \label{def:sensitivity}
A query $f$ has global sensitivity $GS(f)$, where:
\begin{equation}\label{eq:sensitivity}
GS(f) = \max_{x,y:||x-y||_1\leq 1}||f(x)-f(y)||_1 \enspace .
\end{equation}
\end{definition}
Our proposed algorithm uses a more sophisticated form of sensitivity, known as smooth sensitivity, described below in Section~\ref{subsec:smooth}.

\subsection{Smooth Sensitivity}\label{subsec:smooth}

The global sensitivity of a function $f$ is the theoretically largest difference between the output of $f(x)$ and $f(y)$, for any possible dataset $x$ and its neighbor $y$. Instead, we can consider the local sensitivity of $f$, which takes into account a specific $x$:

\begin{definition}[Local Sensitivity \citep{Nissim2007a}]\label{def:local}
For $f:D^n \rightarrow \mathbb{R}^d$ where $n,d\in\mathbb{N}$, and data $x\in D^n$, the local sensitivity of $f$ at $x$ (with respect to the $\ell_1$ metric) is
\begin{equation}
LS_f(x) = \max_{y:||x-y||_1\leq 1}||f(x)-f(y)||_1 \enspace .
\end{equation}
\end{definition}
Unfortunately this definition of sensitivity is not differentially private. Nissim et al. developed a method for making it differentially private in 2007, dubbing it smooth sensitivity:
\begin{definition}[Smooth Sensitivity \citep{Nissim2007a}]\label{def:smooth}
The local sensitivity of $f$, with distance $k$ between datasets $x$ and $y$, is 
\begin{equation}
S^k(x) = \max_{y:||x-y||_1\leq k} LS_f(y) \enspace .
\end{equation}

The smooth sensitivity of $f$ can now be expressed using $S^k(x)$:
\begin{equation}\label{eq:smooth}
S^*(f, x) = \max_{k=0,1,...,n} e^{-k\epsilon}S^k(x)
\end{equation}
where $\epsilon$ is the privacy budget of $f$.
\end{definition}

Smooth sensitivity allows for much less noise to be added while still achieving differential privacy, by analyzing the actual dataset $x$ instead of just assuming the worst-case scenario.

\subsection{Random Decision Trees}\label{subsec:RDTs}

Our proposed technique involves building an ensemble of random decision trees. We briefly go over the basics of random decision trees \citep{Fan2003} here. Decision tree algorithms are a non-parametric supervised learning method used for classification \citep{Han2006a}. They make no assumptions about the distribution of the underlying data, and are trained on known labeled data to correctly classify previously unknown data. The algorithm is what builds the tree, and the resulting tree then acts as a model for making predictions.

Fig.~\ref{fig:tree} is an example of a decision tree. A decision tree is an acyclic directed graph, built using top-down recursive partitioning of the training data. The records in the dataset are recursively divided into subsets using ``tests'' in each node of the tree. The test checks the value each record has for a chosen feature, sending records to the child node that matches the value they have. 

\begin{figure}[t]
	\centering
	\includegraphics[width=2.0in]{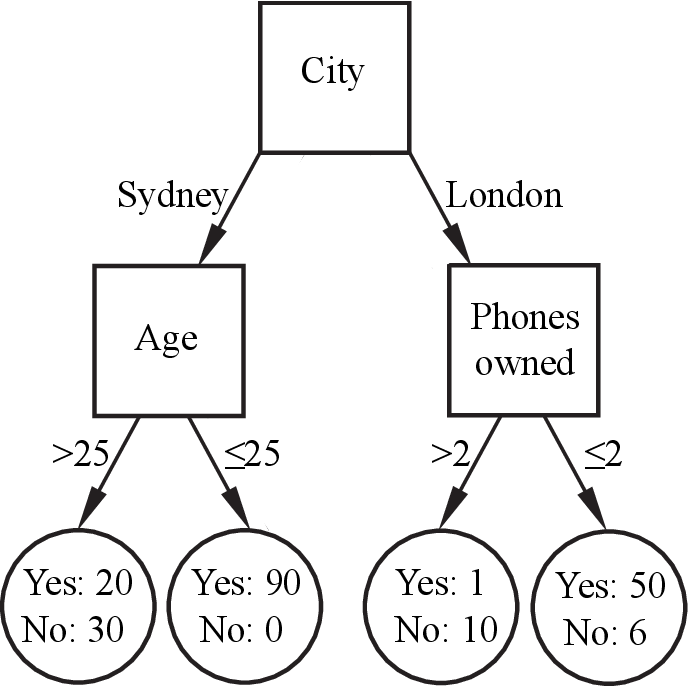}
	\caption{An example of a decision tree, with both discrete and continuous features, a depth of $d=3$, and a binary (``Yes/No'') class label.}
	\label{fig:tree}
\end{figure}

\subsection{Experiment Methodology}\label{subsec:methodology}

We present empirical results throughout the paper, demonstrating how the theory of our algorithm performs in practice. All our experiments are repeated ten times, with each test using ten-fold cross validation, for a total of 100 results that are then averaged. We include one standard deviation when presenting the average result. The majority of our experiments use a privacy budget of $\epsilon=1$. We also perform some experiments with $\epsilon=0.01,0.1,0.2$. In real-world scenarios, the privacy budget given to a user very much depends on that specific case; while a value of $\epsilon=0.01$ is sometimes suggested \citep{Dwork2014,Dwork2008}, values as high as $\epsilon=8.6$ have been used in large, public projects \citep{Machanavajjhala2008}. We also demonstrate in Section~\ref{subsec:size} how the larger the size $n$ of a dataset, the smaller $\epsilon$ can be without injecting too much noise into the resulting trees; a phenomena well understood when implementing differential privacy \citep{Dwork2014,Dwork2008,Dwork2011}.

\begin{table}[t]
\renewcommand{\arraystretch}{1.3}
\centering
\footnotesize
	\begin{tabular}{>{\centering\arraybackslash}m{4.0cm} >{\centering\arraybackslash}m{2.5cm} >{\centering\arraybackslash}m{2.5cm}}
	\noalign{\smallskip}\hline\noalign{\smallskip}	
	\textbf{Dataset Name} & \textbf{n\_informative} & \textbf{n\_random} \\
	\noalign{\smallskip}\hline\noalign{\smallskip}
	SynthA & 5 & 0 \\
	SynthB & 10 & 0 \\
	SynthC & 15 & 0 \\
	SynthD & 10 & 5 \\
	SynthE & 5 & 10 \\
	SynthF & 5 & 5 \\
	SynthG & 10 & 10 \\
	\noalign{\smallskip}\hline
  \end{tabular}
\caption{Parameters for the synthetic datasets we use throughout the paper.}
\label{tab:synthetic}
\end{table}

In our experiments involving real-world data, the data was collected from the UCI Machine Learning Repository \citep{Bache2013}. Details of these datasets are presented in Table~\ref{tab:depths_examples}; other details are available on the UCI website. Some of the real-world datasets we use have more than two class labels; PenWritten has ten class labels, and WallSensor, Nursery and Claves have four. One of the advantages of our findings is that they are generalized to the multi-label (i.e., non-binary) case, and so we include that case in our testing. In experiments involving synthetic data, the data are generated with the \emph{sci-kit learn} package in Python \citep{Pedregosa2011}. The parameters that differentiate each synthetic dataset are defined by the \emph{sklearn.datasets.make\_classification} function. We generate $n=30,000$ records for each dataset, with different numbers of continuous features $m$, and use a balanced binary class feature. The parameters using non-default settings are defined as follows:

\textbf{n\_informative:} ``The number of informative features. Each class is composed of a number of Gaussian clusters each located around the vertices of a hypercube in a subspace of dimension \emph{n\_informative}. For each cluster, informative features are drawn independently from $N(0,1)$ and then randomly linearly combined [into records] within each cluster in order to add covariance.'' \citep{Pedregosa2011}

\textbf{n\_random:} The number of useless features, generated randomly.

The above parameters sum to the total number of features $m$ generated for a dataset. Full details about the parameters can be found in the online documentation for \emph{sklearn.datasets.make\_classification}. Table~\ref{tab:synthetic} provides the parameter values we use for the seven synthetic datasets used in our experiments.

In some experiments we compare our differentially private algorithm to ``non-private'' algorithms. By ``non-private'' we mean classification algorithms that do not attempt to protect privacy in any way -- they simply try to maximize prediction accuracy as much as possible. We present results for two non-private algorithms in this paper; \citet{Breiman2001}'s Random Forest algorithm, and a version of our proposed algorithm with the privacy protection removed. More specifically, the non-private version of our proposed algorithm is exactly the same as the private version except for the following: it uses a privacy budget of $\epsilon=\infty$ (i.e., it always outputs the correct answer); and it filters all the data through every tree instead of a subset of data (explained further in Section~\ref{subsec:disjoint}). Note that our aim is not to beat these non-private techniques (which is all but impossible), but simply to use them as a reference point for how high prediction accuracy can realistically get for the datasets we use.

\section{Related Work}\label{sec:related}

Performing data mining while preserving the privacy of the individuals described in the data has been a topic of much research for nearly 20 years \citep{Brankovic1999}. Many techniques have been developed to preserve privacy, such as $k$-anonymity \citep{Sweeney2002} and noise addition \citep{Agrawal2000,Yun2015}. Decision tree algorithms have been proposed in the past for both these types of privacy, such as \citet{Fung2005,Fung2007} for k-anonymity and \citet{Islam2011,Fletcher2015a} for noise addition, all aiming to keep the quality of the data mining results as high as possible \citep{Fletcher2014}. 

In 2006 a new paradigm was proposed: differential privacy \citep{Dwork2006}. Many data mining algorithms have since been developed that harness differential privacy \citep{Ji2014,Sarwate2013,Zhu2014}. Applications range from releasing private data using decision trees \citep{Mohammed2011}, to kernel SVM learning \citep{Jain2013}, to clustering \citep{Chen2015}. A querying language named PINQ has also been designed by Frank McSherry at Microsoft, allowing users to query a dataset in similar way to SQL, except that the outputs are differentially-private with no privacy expertise required from the user \citep{McSherry2009}. In the following section, we focus in on the differentially-private applications closest to our proposed technique.

\subsection{Differentially Private Random Decision Trees} \label{subsec:DP-RDTs}

Several differentially-private decision tree algorithms have been proposed in recent years \citep{Friedman2010,Jagannathan2012,Fletcher2015b,Fletcher2015c,Rana2016}. Of these, two of them took a similar approach to our paper and used random decision trees to construct a decision forest \citep{Jagannathan2012,Fletcher2015c}. The other three make more traditional decision trees, using greedy heuristics in each node to construct non-random trees \citep{Friedman2010,Fletcher2015b,Rana2016}. All of them achieve differential privacy by adding Laplace noise \citep{Dwork2006a} to the counts of the labels in the nodes.\footnote{The greedy decision tree algorithms also use the Exponential Mechanism to add noise when choosing features in each node.} While this approach makes good use of parallel composition (Definition~\ref{def:parallel}), it scales poorly with multiple labels, since noise needs to be independently added to each label count.

One of the disadvantages of using a splitting criteria in each node is that the user must query the data to do so. This is an expenditure of the privacy budget that random decision trees avoid entirely. Random decision trees have been shown to be competitive against greedy decision trees even in non-private scenarios \citep{Fan2003,Geurts2006}, and especially appropriate in our private scenario. Empirically, this intuition has been validated \citep{Jagannathan2012,Fletcher2015c}, with better prediction accuracy being achieved over their greedy counterparts \citep{Friedman2010,Fletcher2015b,Rana2016}. \citet{Jagannathan2012} improved prediction accuracy by taking advantage of some combinatorial reasoning developed by \citet{Fan2003}, whereby the optimal tree depth is equal to half the number of features, $m/2$. \citet{Fan2003} also empirically demonstrated that 10 to 30 trees is often enough to achieve most of the possible prediction accuracy potential, with increases in accuracy flattening out beyond 30 trees.

However, none of these techniques can take advantage of smooth sensitivity, and thus add much more noise than is necessary to build a differentially-privacy classifier. In fact, the amount of Laplace noise that needs to be added to a frequency query (such as label counts) cannot be reduced by using smooth sensitivity instead of global sensitivity; adding or removing one record can always change a count by 1, even when considering a specific dataset $x$. What this means is that submitting a frequency query to the dataset is an inherently expensive query, and should be avoided in favor of less expensive queries if possible. We demonstrate that this is indeed possible by devising a less sensitivity query that outputs a similar answer, resulting in less noise overall.

Aside from our proposed algorithm and \citet{Rana2016}, none of the above decision tree algorithms tested their algorithms with continuous features. However \citet{Friedman2010} includes a technically correct but costly (in terms of privacy) extension for handling continuous features, and \citet{Jagannathan2012} states that their algorithm can be trivially extended to handle continuous features by uniformly randomly selecting a split point from the continuous feature's domain. Randomly selecting a split point is the same approach used by non-private random decision trees, and is the same approach we implement for our proposed algorithm. We compare our algorithm to \citet{Friedman2010}, \citet{Jagannathan2012} and \citet{Rana2016} in our experiments, including the extensions for continuous features.

\section{Our Decision Forest Algorithm} \label{sec:forest}

We propose a decision forest algorithm that is tailored to the differential privacy scenario. The complete algorithm can be seen in Algorithm~\ref{algorithm} and Algorithm~\ref{algorithm2}. We summarize our algorithm with the following steps:

\begin{description}
\item[Step 1] Calculate the optimal tree depth (Section~\ref{subsec:depth} and Line~\ref{alg:d} of Algorithm~\ref{algorithm}).
\item[Step 2] Decide how many trees $\tau$ to build (Section~\ref{subsec:numtrees} and user input in Algorithm~\ref{algorithm}). 
\item[Step 3] Build a forest (Lines~\ref{alg:ensemble}--\ref{alg:ensemble2} of Algorithm~\ref{algorithm}) of $\tau$ random decision trees, without needing to query the data (Section~\ref{subsec:RDTs} and the \texttt{BuildTree} procedure in Algorithm~\ref{algorithm2}).
\item[Step 4] Query the leaf nodes and output the majority class labels (Section~\ref{subsec:majority}, Line~\ref{alg:expo} of Algorithm~\ref{algorithm2}).
\item[Output] The decision forest model has finished being built. It can now be used to classify the labels of unseen data (Section~\ref{subsec:RDTs}).
\end{description}

\begin{algorithm}[t]
\small
\caption{\small The proposed Differentially Private Random Decision Forest with Smooth Sensitivity.}
\label{algorithm}
\begin{algorithmic}[1]
\Procedure{BuildForest}{Privacy budget $\epsilon$, dataset $x$, number of trees $\tau$, set of continuous features $S$, set of discrete features $R$, class label $C$}
\State $F \leftarrow \{\}$
\State $d \leftarrow$ The optimal tree depth according to our \textbf{Theorem~\ref{theorem_depth}}. \label{alg:d}
\For{~$t = 1,\ldots ,\tau$} \label{alg:ensemble}
	\State $T \leftarrow $ \Call{BuildTree}{$d$, 0, $S$, $R$}
	\State $F \leftarrow F \cup T$
\EndFor \label{alg:ensemble2}
\State $F \leftarrow$ \Call{GetMajorityLabels}{$\epsilon$, $x$, $C$, $F$}
\State\Return $F$
\EndProcedure

\end{algorithmic}
\end{algorithm}

\begin{algorithm} 
\footnotesize
\caption{\small A continuation of the algorithm presented in Algorithm~\ref{algorithm}.}
\label{algorithm2}
\begin{algorithmic}[1]

\Procedure{BuildTree}{Maximum tree depth $d$, current depth $d'$, continuous features $S$, discrete features $R$}
\State $T \leftarrow \{\}$
\If{~$d'<d$} \Comment{Termination criteria.}
	\State Uniformly randomly select a feature $g$ from $S \cup R$ to split the current node with.
	\If{~$g \in S$}
		\State Uniformly randomly select a splitting point $p$ within the current domain of $g$.
		\State Original domain of $g~\leftarrow~$Update the lower bound to $p$.
		\State $T \leftarrow T~\cup$ \Call{BuildTree}{$d$, $d'+1$, $S$, $R$} \Comment{Left child node.}
		\State Original domain of $g~\leftarrow~$Update the upper bound to $p$.
		\State $T \leftarrow T~\cup$ \Call{BuildTree}{$d$, $d'+1$, $S$, $R$} \Comment{Right child node.}
	\Else{~$g \in R$}
		\State $B \leftarrow B - g$ \Comment{Discrete features can only be selected once in a root-to-leaf path.}
		\ForAll{~$g_i\in g$}
			\State $T \leftarrow T~\cup$ \Call{BuildTree}{$d$, $d'+1$, $S$, $R$} \Comment{All child nodes.}
		\EndFor
	\EndIf
\EndIf
\State\Return $T$
\EndProcedure

\Procedure{GetMajorityLabels}{Privacy budget $\epsilon$, dataset $x$, class label $C$, forest $F$} 
\For{~$t = 1,\ldots ,\tau$}
	\State Query $f(x_t) \leftarrow$ Scope query $f$ to $x_t$ where $x_t$ is a disjoint subset of $x=\left\{x_1,...,x_\tau\right\}$
	\ForAll{Leaf nodes $L=\left\{L_i,\forall i\right\}$ in tree $F_t$, $F=\left\{F_1,...,F_\tau\right\}$}
		\State $f(x_t\cap L_i) \leftarrow$ Scope $f(x_t)$ with all the tests in the root-to-leaf path leading to leaf node $L_i$.
		\State Majority label of $L_i \leftarrow c\in C$ outputted by the Exponential Mechanism using $f(x_t\cap L_i)$ and $\epsilon$, as well as the scoring function and smooth sensitivity proposed in \textbf{Theorem~\ref{theorem}}. \label{alg:expo}
	\EndFor
\EndFor
\State\Return $F$
\EndProcedure
\end{algorithmic}
\end{algorithm}

The construction of the random decision trees is the same as the conventional approach \citep{Fan2003,Geurts2006}; features are randomly chosen for each node. The \texttt{BuildTree} procedure in Algorithm~\ref{algorithm2} outlines the recursive tree-building process. The novel parts of our algorithm are the following: how we output the majority (i.e., most frequent) label of each leaf node (Section~\ref{subsec:majority}); our efficient utilization of the privacy budget (Section~\ref{subsec:disjoint}); our proposed tree depth, extending the non-private work of \citet{Fan2003} to handle numerical features (Section~\ref{subsec:depth}); and the number of trees we build (Section~\ref{subsec:numtrees}).

Observe that in Step 4, the query constructed for each leaf node includes the rules (i.e., tests) used by the nodes above the leaf node leading back to the root node. This is how a differentially-private algorithm can achieve the ``filtering'' process described in Section~\ref{subsec:RDTs}; since we can only access the data via queries, we cannot hold the whole dataset in memory for the root node, and partition it down through the tree. Instead, each leaf node's query narrows the parameters of the what training records are included so that they match the tests in the root-to-leaf path. This is an engineering consideration more than anything (it is less computationally efficient for example), with no effect on the algorithm's results, but is worth mentioning nonetheless.

Differential privacy is achieved by our algorithm only outputting the following: the structure of the trees (which does not use the data); and the most frequent label in each leaf node, which is done using the Exponential Mechanism. The user can then use the outputted labels from the leaf nodes in whatever way they wish; differential privacy is immune from post-processing, and differentially-private outputs can never incur additional privacy costs after the fact \citep{Dwork2014}. More specifically, the user is free to use majority voting; they can predict the label of a new record using the most common predicted label from all the trees in the ensemble.

\subsection{Outputting the Majority Label}\label{subsec:majority}

When outputting details about leaf nodes, rather than trying to return approximately-correct class frequencies like in \citet{Jagannathan2012} and \citet{Fletcher2015c}, we instead observe that this is more information than is necessary in order to make a highly predictive classifier. To predict the class label of future records, we only need to know the majority (i.e., most frequent) class label in each leaf node (regardless of the number of times that label occurred). By not ``wasting'' some of the privacy budget on information we do not need, we propose querying the data with the Exponential Mechanism to only output the (discrete) class label that is most frequent.

Definition~\ref{def:expo} describes how the Exponential Mechanism is capable of returning the discrete output of the most frequent class label in a leaf $x$. Our query takes the same form as seen in the definition, where the query $f$ will output the most frequently occurring label in $x$ with high probability. The precise probability is dependent on the scoring function $u$ and the privacy budget $\epsilon$. We propose the following novel scoring function:

\begin{theorem}\label{theorem}
Given the leaf $x$ of a decision tree, the most frequent label can be differentially-privately queried with the Exponential Mechanism \citep{McSherry2007}. The scoring function used for this query can be the piecewise linear function:
\begin{equation}\label{eq:score}
u(c, x) = 
	\begin{cases} 
      1 & c = \argmax_{i\in C} n_i \\
      0 & \mbox{otherwise}
   \end{cases}
	\enspace ,
\end{equation}
where $n_c$ is the number of occurrences of $c$ in $x$. Each class label $c\in C$ will have a score: 1 if the label is the most frequently occurring label in the leaf; 0 otherwise.

The smooth sensitivity of $u(c, x)$ is
\begin{equation}\label{eq:score_smooth}
S^*(u, x) = e^{-j\epsilon}
\end{equation}
where $\epsilon$ is the privacy budget of the query and $j$ equals the difference between the most frequent and the second-most frequent labels in $x$, $n_{c_1}-n_{c_2}$.
\end{theorem}

\begin{proof}
The global sensitivity of the scoring function $u$ seen in (\ref{eq:score}) is 1 because, when considering any possible neighbors $x$ and $y$, adding or removing any record has the potential to change which label(s) occurs most frequently:
\begin{equation}\label{eq:score_GS}
GS(u) = \max_{x,y:||x-y||_1\geq 1}\max_{c\in C}||u(c,x)-u(c,y)||_1 = 1 \enspace .
\end{equation}
A label will change from a score of 0 to 1 when it appears equally as frequently as the original most frequent label (at this point, 2 labels will be reporting a score of 1). A label's score will change from 1 to 0 when another label occurs more frequently. In either case, the most that the scoring function $u$ can change by is $GS(u)=1$.

The local sensitivity of $u$ differs from the global sensitivity by taking into account a specific $x$, rather than considering the worst theoretical outcome over all possible $x$'s. From Definition~\ref{def:local} and Definition~\ref{def:smooth} we get
\begin{equation}
S^k(x) = \max_{y:||x-y||_1\leq k}\max_{c\in C} ||u(c,x)-u(c,y)||_1 \enspace .
\end{equation}
Because we are considering the actual frequencies of each label in $x$, we are interested in how many records would need to be added or removed before a different label achieves a score $u$ of 1. This occurs when enough records with a different label have been added to equal the most frequent label in $x$, or when enough records have been removed from $x$ to drop the most frequent label to the same frequency as the next most common label. In either case, this first will occur at a distance $j$ from $x$, where $j$ is the difference between the most frequent and second-most frequent labels in $x$. Until then, for $k<j$, $S^k(x)=0$. When $k=j$, it is possible to have two labels with a score of $u=1$. When $k>j$, it is possible to have one or more labels with a score of $u=1$. Regardless, for $k\geq j$, the local sensitivity is $S^k(x)=1$. To summarize, the local sensitivity of our scoring function $u$ is 0 until $y$ is sufficiently far away from $x$, at which point the sensitivity becomes 1. We can represent this as
\begin{equation}\label{eq:local_possibilities}
S^k(x) = 
	\begin{cases} 
      0 & k<j \\
			1 & k\geq j
	\end{cases}
; j=n_{c_1}-n_{c_2}, k\in\mathbb{N} \enspace ,
\end{equation}
where $n_{c_1}$ and $n_{c_2}$ are the frequencies of the most common and second-most common labels in $x$, respectively.

To finish implementing Definition~\ref{def:smooth}, we input $S^k(x)$ from (\ref{eq:local_possibilities}) into the smooth sensitivity (\ref{eq:smooth}). We now find the value of $k$ for which $e^{-k\epsilon}S^k(x)$ is maximized. For $k<j$, we have
\begin{equation}
e^{-k\epsilon}S^{k}(x) = 0 \enspace .
\end{equation}
Two other possible scenarios exist: $k=j$ and $k>j$. Because $S^k(x)$ is never larger than 1, and $e^{-k\epsilon}$ becomes smaller as $k$ gets larger, we can deduce that $e^{-k\epsilon}S^k(x)$ is largest when $k=j$:
\begin{equation}
S^*(u, x) = \max_{k=0,1,...,n} e^{-k\epsilon}S^k(x) = e^{-j\epsilon} \enspace .
\end{equation}
\end{proof}

The above proof holds in any scenario where the most frequent item in a set is to be outputted. It also holds if the \emph{least} frequent item is the desired output.

\begin{table}
\centering
\footnotesize
\renewcommand{\arraystretch}{1.3}
	\begin{tabular}{>{\centering\arraybackslash}m{0.7cm} >{\centering\arraybackslash}m{3.0cm} >{\centering\arraybackslash}m{3.0cm} >{\centering\arraybackslash}m{3.0cm}}
	\noalign{\smallskip}\hline\noalign{\smallskip}	
	$j$ & $S^*(u, x)$ when $\epsilon=0.01$ & $S^*(u, x)$ when $\epsilon=0.1$ & $S^*(u, x)$ when $\epsilon=1.0$ \\
	\noalign{\smallskip}\hline\noalign{\smallskip}
	0 & 1.00000 & 1.00000 & 1.00000 \\ 
	1 & 0.99004 & 0.90483 & 0.36788 \\
	5 & 0.95122 & 0.60653 & 0.00674 \\
	10 & 0.90483 & 0.36788 & 0.00005 \\
	50 & 0.60653 & 0.00674 & $1.93\times 10^{-22}$ \\
	100 & 0.36788 & 0.00005 & $3.72\times 10^{-44}$ \\
	500 & 0.00674 & $1.93\times 10^{-22}$ & $7.12\times 10^{-218}$ \\
	\noalign{\smallskip}\hline
  \end{tabular}
\caption{Example smooth sensitivities of our scoring function $u$ using Theorem~\ref{theorem}, when $\epsilon=0.01, 0.1, 1.0$. Note how $j$ and $\epsilon$ have an equal impact on the result -- increasing $j$ tenfold is the same as increasing $\epsilon$ tenfold.}
\label{tab:smooth_examples}
\end{table}

Using our proposed smooth sensitivity instead of the global sensitivity of 1, we can guarantee that equal or less noise is added to all queries. In some cases, the noise can be substantially lower. The smooth sensitivity of a query on data $x$ to return the most frequent label is dependent on only $j$ and $\epsilon$. We present the worst-case scenario (i.e., when $j=0$) as well as several other scenarios in Table~\ref{tab:smooth_examples}, using $\epsilon=0.01$, $\epsilon=0.1$ and $\epsilon=1.0$.

In Fig.~\ref{fig:smooth_vs_global}, we present the practical effect of our smooth sensitivity. We present the prediction accuracy of our differentially-private random decision forest on seven synthetic datasets (described in Section~\ref{sec:experiments}). Two scenarios are demonstrated in the figure; when the Exponential Mechanism queries in the leaf nodes use smooth sensitivity, and when they use global sensitivity. All our other parameters, described in the following sections, use the default settings.

\begin{figure}[t]
	\centering
	{\footnotesize \def\svgwidth{4.0in} 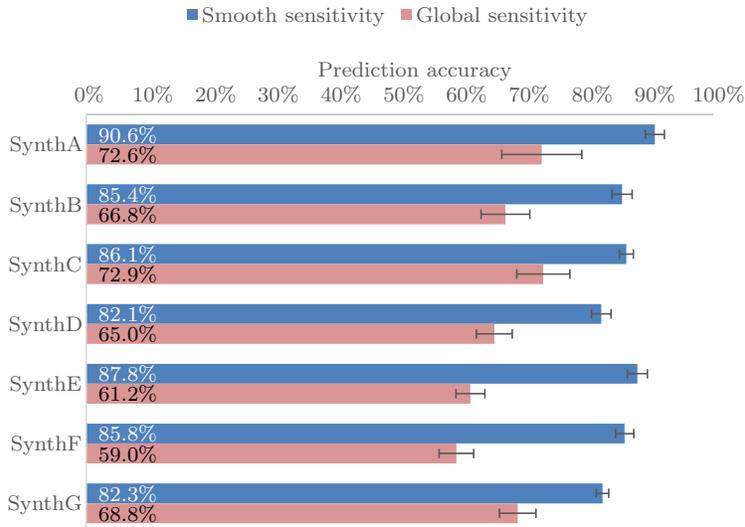}
	\caption{The average prediction accuracy of two possible versions of our technique with a privacy budget of $\epsilon=1$. One version uses our proposed smooth sensitivity of the scoring function in the Exponential Mechanism. The other version uses the global sensitivity of the scoring function (i.e., $\Delta(u)=1$).}
	\label{fig:smooth_vs_global}
\end{figure}

Note the substantial improvement in prediction accuracy when using our smooth sensitivity -- up to 26 percentage points in some cases, and never worse. This matches what we expect from the theory, where even in the worse case scenario, the smooth sensitivity is equal to the global sensitivity. Fig.~\ref{fig:smooth_vs_global} presents the results for $\epsilon=1$, where Table~\ref{tab:smooth_examples} tells us that even with a modest $j$ of 10, the exponent of the Exponential Mechanism is improved 20,000-fold. For $\epsilon=0.1$, the improvement would still be three-fold.

\subsection{Using Disjoint Data}\label{subsec:disjoint}

When building multiple trees (say $\tau$ trees), there are two fundamentally different ways we can use our privacy budget $\epsilon$. One is to use composition (Definition~\ref{def:composition}), where all the records in $x$ ($n=|x|$) are used in every tree and we divide $\epsilon$ evenly amongst the trees, $\epsilon'=\epsilon/\tau$. The other way is to use parallel composition (Definition~\ref{def:parallel}), dividing $x$ into disjoint subsets\footnote{When we refer to disjoint subsets of data, we are referring to a collection of records that are sampled from the total dataset without replacement. Each record can only appear in one disjoint subset.} evenly amongst the trees ($n=|x|/\tau$) and using the entire $\epsilon$ budget in each tree.

When deciding which method to use, let us consider the factors that affect how noisy the output of our query proposed in Theorem~\ref{theorem} is. There are two such factors: $\epsilon$; and the sensitivity $S^*(u, x)$, which is itself dependent on only $\epsilon$ and $j$. When $n$ is larger, each leaf node will contain more records on average. Since we assume all the records in $x$ are sampled from the same population, the multinomial distribution will be approximately the same for both $n=|x|$ and $n=|x|/\tau$ number of records. Thus the relative frequency of both the number of records, and the ratio of class labels, will be the same in each leaf node, and $j$ will change accordingly. We can write this as $\mathbf{E}[j]\propto n$.

One ramification of using all of $x$ in each tree is that the smaller $\epsilon/\tau$ privacy budget affects both the numerator and denominator of the Exponential Mechanism, unlike $j$ which only affects the denominator. This means that $\epsilon$ has a larger impact on the result than $j$, and good-scoring labels have a higher chance of being outputted by the Exponential Mechanism when disjoint subsets of $x$ are used:

\begin{equation}\label{eq:disjoint}
\exp\left(\frac{\epsilon\times u(z,x)}{2\exp(-j/\tau\times \epsilon)}\right) > 
\exp\left(\frac{\epsilon/\tau\times u(z,x)}{2\exp(-j\times \epsilon/\tau)}\right)
\end{equation}
when $\mathbf{E}[j]\propto n$. This effect is most prominent with smaller datasets where $j$ is small, before the exponential nature of the smooth sensitivity overpowers the numerator by more than two or three orders of magnitude.

\begin{figure}[t]
	\centering
	\subfloat[Prediction accuracy]{\footnotesize \def\svgwidth{2.49in} 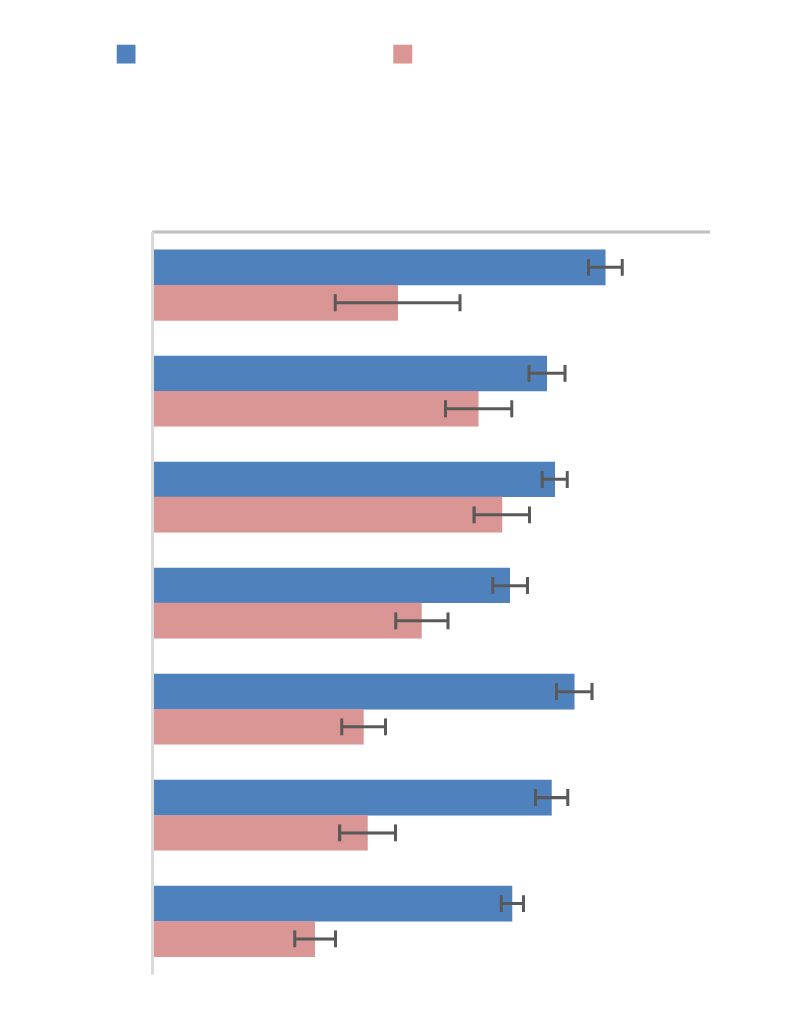\label{fig:disjoint_a}}
	\subfloat[Smooth sensitivity]{\footnotesize \def\svgwidth{2.13in} 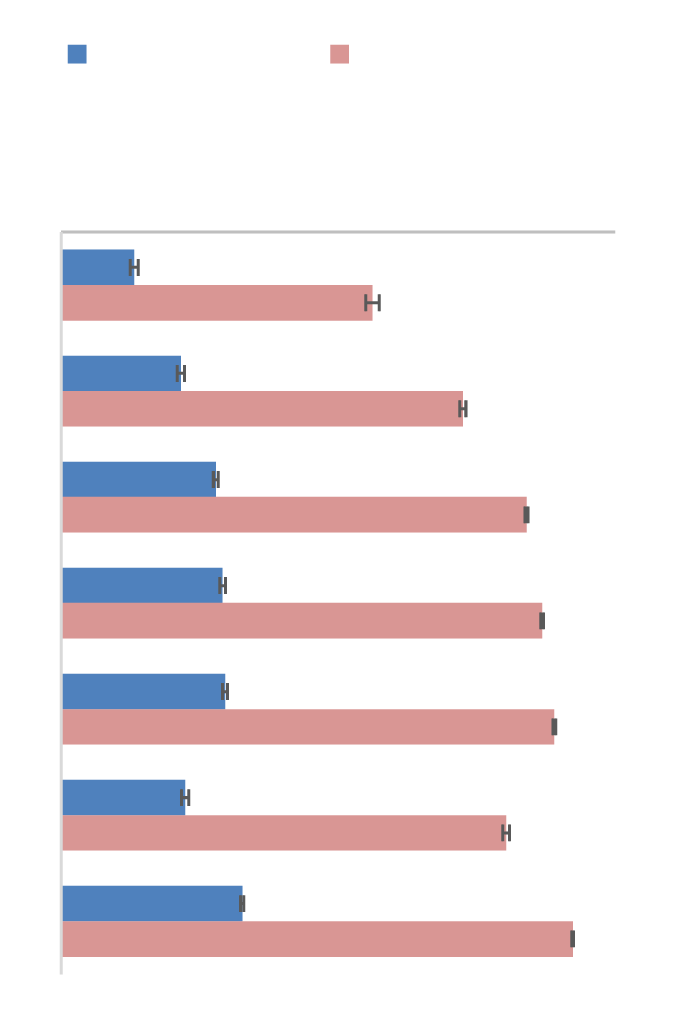\label{fig:disjoint_b}}
	\caption{The (a) average prediction accuracy and (b) average smooth sensitivity of our proposed algorithm, with and without using disjoint data in each decision tree. All other parameters remain constant, using the default settings described in the other sections. The budget is $\epsilon=1$ and the number of trees is $\tau=100$. Recall that lower sensitivity is better. Due to space constraints, the horizontal axis in (a) starts at 50\%.}
	\label{fig:disjoint}
\end{figure}

Another factor to consider when choosing between the composition and parallel composition theorems is the fact that the correlation between $j$ and $n$ is not one-to-one. If it was, that would imply that all leaf nodes that are empty (i.e., have no records in them) when $n=|x|/\tau$ will remain empty when $n=|x|$, when instead the reality is that the sample size was simply not big enough for any records to be in some of the leaf nodes. Therefore we expect that for some leaf nodes, neither their support nor their class label proportions (and in turn, $j$) will increase linearly with $n$. Of course, some previously empty leaf nodes now have records in them, and the most frequent label outputted by the Exponential Mechanism in these leaf nodes will no longer be purely random. The privacy budget is only $\epsilon/\tau$ in this scenario though, leading to the difference between ``purely random'' and ``almost purely random'' being trivially small for some of these leafs. We empirically test the correlation between $j$ and $n$, as well as the number of empty leaf nodes, in Section~\ref{subsec:leafs}.

Due to the above factors, we propose using parallel composition, and using disjoint data in each tree with the full privacy budget. Fig.~\ref{fig:disjoint_a} empirically demonstrates the improved prediction accuracy we see due to this decision. The only parameter changed for the comparison is whether or not discrete subsets of data are used in each of the 100 trees (and if not, $\epsilon$ is divided instead). All other parameters use the default settings, described in the other sections. This decision has the added benefit of substantially decreasing the computation time of our algorithm, from $O(n\log n)$ to $O(\frac{n}{\tau}\log\frac{n}{\tau})$, where $n$ is the number of records in the dataset.

We see improvements of up to 18 percentage points when disjoint data is used, and no losses by more than a fraction of one standard deviation. Aside from the large improvements in prediction accuracy, we can make another observation from Fig.~\ref{fig:disjoint_b}; one that might be surprising at first. The average\footnote{Note that we don't include empty leaf nodes (which have a smooth sensitivity of 1) when calculating the average smooth sensitivity, since 70 to 99.9\% of the leaf nodes are usually empty.} smooth sensitivity (that is, $\exp(-j\epsilon)$) is substantially better when disjoint data is used, when initial intuition might tell us that $j$ and $\epsilon$ should be offsetting one another somewhat equally. Instead, we find that using all of the data in every tree does not lead to an increase in the average $j$ that is as large as the decrease in the privacy budget when going from $\epsilon$ to $\epsilon/\tau$. We explore this phenomena further in Section~\ref{subsec:leafs}.

\subsubsection{Empty Leaf Nodes}\label{subsec:leafs}

For almost any (non-uniform) distribution of feature values, records will start clumping together in certain nodes, with other nodes receiving very few records. This clumping becomes exacerbated for each non-uniformly distributed feature tested in each level of a tree. As a tree grows larger, the chances of some leafs having zero records in them also grows larger. We refer to leaf nodes with zero records in them as ``empty''. We use a simple example to demonstrate:

\begin{example}
Let us imagine we have a dataset made of $m$ features and $n$ records, where each features has three discrete values that follow the normal distribution: value $v_1$ contains 68\% of the records; value $v_2$ contains $95\%-68\%=27\%$; and value $v_3$ has the remaining $100\%-95\%=5\%$ of the records. Let us further assume that all $m$ features are independent, to simplify the simulation. If we were to build a tree of depth $m/2$ with these $m$ features, we would have a tree with $3^{m/2}$ leaf nodes. Out of these leaf nodes, a single leaf would contain $0.68^{m/2}\times n$ records, and ${{m/2}\choose 2}-1$ other leaf nodes would contain proportions of records that were multiplications of 0.68 and 0.27 (such as $0.68\times 0.68\times 0.27\times \ldots$). Conversely, there would be a leaf with $0.05^{m/2}\times n$ records in it; using a conservative $m=10$, this would require a dataset of size $n=3,200,000$ for there to be even one record in this leaf.
\end{example}

\begin{table}[t]
\renewcommand{\arraystretch}{1.3}
\centering
\footnotesize
	\begin{tabular}{>{\centering\arraybackslash}m{1.7cm} >{\centering\arraybackslash}m{1.8cm} >{\centering\arraybackslash}m{1.3cm} >{\centering\arraybackslash}m{2.4cm} >{\centering\arraybackslash}m{2.6cm}}
	\noalign{\smallskip}\hline\noalign{\smallskip}	
	\textbf{Dataset} & 
	\begin{tabular}[c]{@{}c@{}}\textbf{Total No.}\\ \textbf{of Records}\end{tabular}  & 
	\begin{tabular}[c]{@{}c@{}}\textbf{Tree}\\ \textbf{Depth}\end{tabular} & 
	\begin{tabular}[c]{@{}c@{}}\textbf{\% of Empty}\\ \textbf{Leafs (1 tree)}\end{tabular} & 
	\begin{tabular}[c]{@{}c@{}}\textbf{\% of Empty}\\ \textbf{Leafs (30 trees)}\end{tabular} \\
	\noalign{\smallskip}\hline\noalign{\smallskip}
	SynthA & 30000 & 5 & $49.8\pm12.8$ & $68.5\pm2.0$ \\ 
	SynthB & 30000 & 8 & $79.1\pm6.3$ & $90.3\pm0.7$ \\ 
	SynthC & 30000 & 12 & $95.0\pm1.5$ & $98.5\pm0.1$ \\ 
	SynthD & 30000 & 12 & $95.7\pm1.3$ & $98.5\pm0.1$ \\ 
	SynthE & 30000 & 12 & $63.3\pm1.8$ & $73.7\pm0.1$ \\ 
	SynthF & 30000 & 8 & $79.2\pm6.2$ & $89.9\pm0.7$ \\ 
	SynthG & 30000 & 15 & $98.7\pm0.4$ & $99.7\pm0.0$ \\ 
	WallSensor & 5456 & 4 & $48.6\pm12.3$ & $65.1\pm2.1$ \\
	PenWritten & 10992 & 12 & $92.4\pm1.2$ & $97.7\pm0.1$ \\
	GammaTele & 19014 & 8 & $86.2\pm3.8$ & $93.5\pm0.5$ \\
	Adult & 30162 & 9 & $99.8\pm0.0$ & $99.9\pm0.0$ \\
	Mushroom & 5644 & 11 & $99.9\pm0.0$ & $99.9\pm0.0$ \\
	Claves & 10800 & 8 & $0.0\pm0.0$ & $28.7\pm0.4$ \\
	Nursery & 12960 & 4 & $0.2\pm.0.5$ & $7.2\pm0.8$ \\
	\noalign{\smallskip}\hline
  \end{tabular}
\caption{The percentage of leaf nodes with no records in them, when building random trees with eight real-world datasets and seven synthetic datasets. We present the results for one tree and 30 trees, when using disjoint data in each tree, and a privacy budget of $\epsilon=1$. We include one standard deviation for each result. The depth of the trees is defined by a novel theorem that we present in Section~\ref{subsec:depth}.}
\label{tab:empty_leafs}
\end{table}

In our experiments, using both synthetic and real-world data, over 85\% of the leaf nodes in any non-trivial tree are usually empty. By virtue of future data being from the same distribution as the training data, however, these empty leaf nodes are unlikely to be visited by future records. Any records that do finish at an empty leaf node will be predicted to have a class label that is randomly chosen with uniform probability (due to all labels having a score of 0 in the Exponential Mechanism). We consider this to be less damaging than the same scenario in Jagannathan's implementation of a differentially private random decision tree \citep{Jagannathan2012}, where labels with a frequency of zero still have Laplace noise added to them. In an empty leaf node in Jagannathan's implementation, where every label reports a purely random frequency (because the true frequencies are zero), highly confident predictions could be falsely created.

Our empirical results are presented in Table~\ref{tab:empty_leafs}. Note how when the data is divided among more trees, the number of empty leaf nodes always increases, as we would expect. This supports the observation we made with Fig.~\ref{fig:disjoint_b}: $j$ does not increase at the same ratio that $n$ does, because many of the extra records are going to previously empty leaf nodes. While this means these previously empty leaf nodes are no longer outputting a most frequent label with pure randomness, Fig.~\ref{fig:disjoint_b} shows us that this reduction in randomness does not negate the increase in randomness in all leaf nodes from dividing the privacy budget into $\epsilon/\tau$.

\subsection{Tree Depth}\label{subsec:depth}

An existing theory about the depth of a random tree was introduced in Section~\ref{subsec:DP-RDTs}, where the ideal depth is half the number of features $m/2$. However this assumes that the features can only be selected \emph{once} in any root-to-leaf path (by virtue of being discrete features, where all values are separated by the first node that selects the feature). Continuous features, on the other hand, can be randomly chosen any number of times (since there are many more split points remaining that could separate the records). The reasoning behind wanting to test $m/2$ unique features in each path still stands, but the probability of this happening at a depth of exactly $m/2$ is much lower if all (or some) of the features are continuous. We propose a new tree depth using the analysis below.

\begin{theorem}\label{theorem_depth}
The expected number $X$ of continuous features $s$ not tested, on any particular root-to-leaf path of depth $d$, is equal to
\begin{equation}\label{eq:expected}
\mathbf{E}[X|d] = s\left(\frac{s-1}{s}\right)^d \enspace ,
\end{equation}
where each tested features is uniformly randomly selected with replacement. Using the same combinatorial reasoning used in \citet{Fan2003} and described in Section~\ref{subsec:DP-RDTs}, the optimal tree depth is therefore
\begin{equation}\label{eq:depth}
d = \argmin_{d:X<s/2} \mathbf{E}[X|d] \enspace .
\end{equation}
\end{theorem}

\begin{table}[t]
\renewcommand{\arraystretch}{1.3}
\centering
\footnotesize
	\begin{tabular}{>{\centering\arraybackslash}m{2.5cm} >{\centering\arraybackslash}m{2.3cm} >{\centering\arraybackslash}m{2.0cm} >{\centering\arraybackslash}m{2.0cm}}
	\noalign{\smallskip}\hline\noalign{\smallskip}	
	\textbf{Dataset} & \textbf{Continuous Features} & \textbf{Discrete Features} & \textbf{Depth} \\
	\noalign{\smallskip}\hline\noalign{\smallskip}
	SynthA & 5 & 0 & 5 \\ 
	SynthB & 10 & 0 & 8 \\ 
	SynthC & 15 & 0 & 12 \\ 
	SynthD & 15 & 0 & 12 \\ 
	SynthE & 15 & 0 & 12 \\ 
	SynthF & 10 & 0 & 8 \\ 
	SynthG & 20 & 0 & 15 \\ 
	WallSensor & 4 & 0 & 4 \\
	PenWritten & 16 & 0 & 12 \\
	GammaTele & 10 & 0 & 8 \\
	Adult & 6 & 8 & 9 \\
	Mushroom & 0 & 22 & 11 \\
	Claves & 0 & 16 & 8 \\
	Nursery & 0 & 8 & 4 \\
	\noalign{\smallskip}\hline
  \end{tabular}
\caption{The depths calculated using Theorem~\ref{theorem_depth} for the synthetic and real datasets used in our experiments.}
\label{tab:depths_examples}
\end{table}

\begin{proof}
Our proof is the same as the proof for estimating the number of empty bins in the ``Balls and Bins'' problem (a companion of the famous ``Birthday Paradox'' and ``Coupon Collector Problem'' \citep{Flajolet2009}). Each continuous feature is a ``bin'', and each node in a root-to-leaf path in a tree is a ``ball''. Each time we randomly select an feature (i.e., throw a ball), it has an equal chance of being (i.e., landing in) any of the $s$ continuous features (i.e., bins).

\begin{figure}[t]
	\centering
	{\footnotesize \def\svgwidth{3.5in} 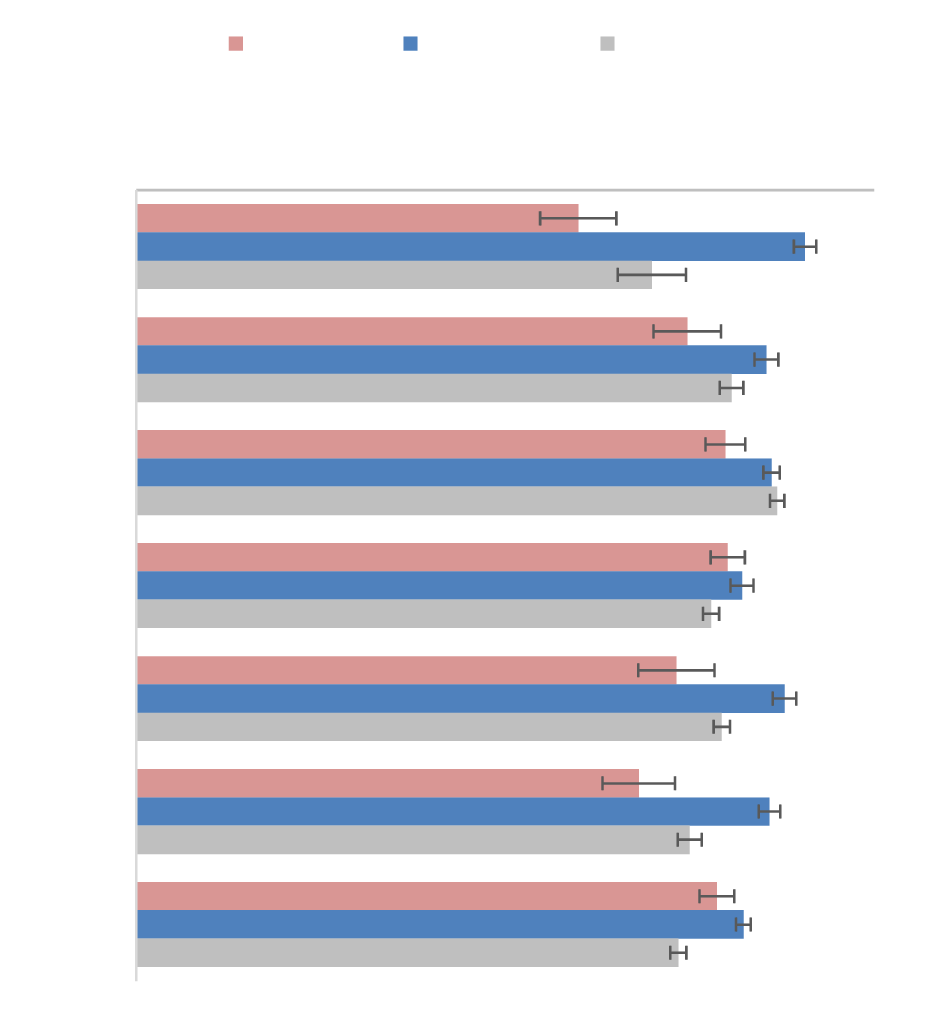}
	\caption{The average prediction accuracy of our proposed algorithm for three different tree depths: $m/2$; our proposed depth defined by Theorem~\ref{theorem_depth}; and $m$. The results are for when the budget is $\epsilon=1$.}
	\label{fig:depth_accuracy}
\end{figure}

Let the random variable $X$ equal the number of features never selected in a root-to-leaf path. For each feature $i$, $X_i$ equals 1 if $i$ is never selected, and 0 otherwise. Since $\mathbf{E}[X] = \mathbf{E}[\sum\limits_i^s X_i] = \sum\limits_i^s \mathbf{E}[X_i]$, we only need to know $\mathbf{E}[X_i]$ to know the expected number of features that are never selected. For any feature $i$, the probability that it will not be selected in a node is equal to the probability that any of the other $s-1$ features will be selected, which equals $\frac{s-1}{s}$. If we repeat this for all $d$ nodes, where each selection is independent of the others, the probability that $i$ is never selected is equal to $(\frac{s-1}{s})^d$. Thus $\mathbf{E}[X_i] = (\frac{s-1}{s})^d$, and $\mathbf{E}[X] = s(\frac{s-1}{s})^d$.
\end{proof}

If a dataset has $r$ discrete features in addition to the $s$ continuous features, we add $r/2$ to the depth $d$ defined by Theorem~\ref{theorem_depth} as per the combinatorial reasoning of \citet{Fan2003}. Table~\ref{tab:depths_examples} shows the tree depths for the datasets we use in our experiments, calculated using Theorem~\ref{theorem_depth}.

Fig.~\ref{fig:depth_accuracy} provides empirical evidence of the prediction accuracy gained when using our proposed depth, rather than simply $m/2$. We also present the prediction accuracy when a depth of $m$ is used, to demonstrate that increasing the depth arbitrarily does not necessarily increase the accuracy (it actually \emph{decreases} the accuracy for all but one dataset). Our proposed depth provides higher prediction accuracy than $m/2$ for all datasets, by at least 2 percentage points and as much as 30 percentage points. For most datasets, we also see less variance in the standard deviation when using our proposed depth over $m/2$, as expected from the analysis done by \citet{Fan2003}.

\subsection{Number of Trees}\label{subsec:numtrees}

Another factor in the implementation of a differentially private decision forest with random decision trees is the number of trees to build. Fig.~\ref{fig:num_trees_acc} and Fig.~\ref{fig:num_trees_flipped} present empirical results for a range of forest sizes (i.e., number of trees). Fig.~\ref{fig:num_trees_acc} shows the average prediction accuracy results for our seven synthetic datasets, for both $\epsilon=0.2$ and $\epsilon=1$.

\begin{figure}
	\centering
	\subfloat[$\epsilon=0.2$]{\footnotesize \def\svgwidth{2.44in} 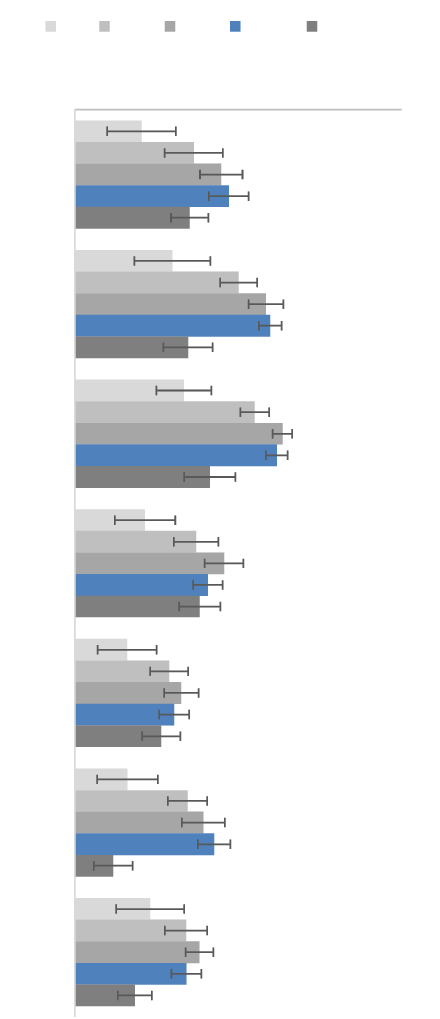\label{fig:num_trees_acc_a}}
	\subfloat[$\epsilon=1.0$]{\footnotesize \def\svgwidth{2.09in} 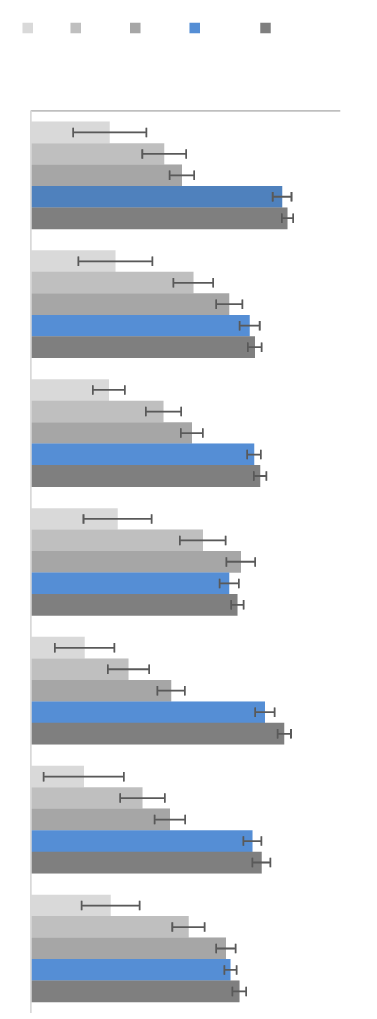\label{fig:num_trees_acc_b}}
	\caption{The average prediction accuracy of our proposed algorithm when building different numbers of trees, with privacy budget (a) $\epsilon=0.2$ and (b) $\epsilon=1.0$. Due to space constraints, the horizontal axis starts at 50\%. We recommend building 100 trees when using our algorithm, seen in blue.}
	\label{fig:num_trees_acc}
\end{figure}

\begin{figure}
	\centering
	{\footnotesize \def\svgwidth{4.0in} 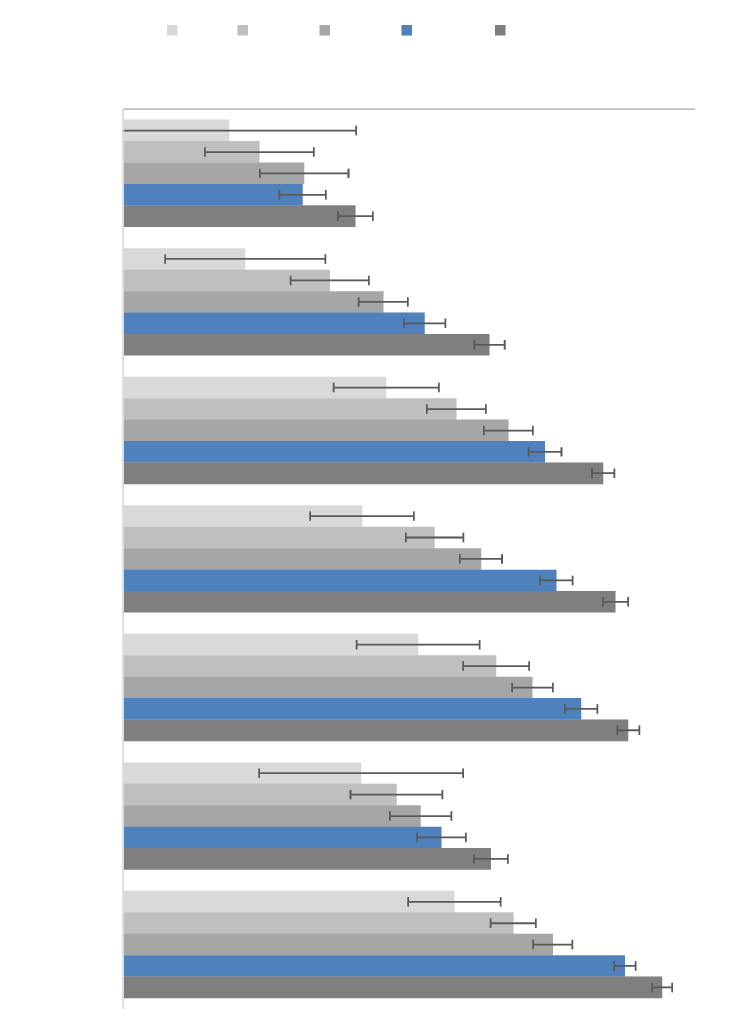}
	\caption{The percentage of non-empty leaf nodes that had their majority (i.e., most frequent) labels changed by the Exponential Mechanism, when building different numbers of trees with a budget of $\epsilon=1$.}
	\label{fig:num_trees_flipped}
\end{figure}

From these results, one observation is that having more trees is not necessarily better. Indeed, for $\epsilon=0.2$ especially, we can see that there appears to be a ``sweet spot'' at 30 to 100 trees where prediction accuracy is highest. As $\epsilon$ increases, this sweet spot increases to 100 to 300 trees. Another observation we can make is that for 1 to 10 trees, the prediction accuracy results vary by a lot more compared to when there are more trees, seen by the larger standard deviations. This is something we expect to see, given the high randomness in the construction of the trees. If we only build one random tree, there is a much higher chance of us getting very lucky or very unlucky when predicting future labels than there is if we build many random trees. With many trees, we can use the predictions made by each tree as votes and select the most voted class label as our prediction. The variance is also reduced due to the disjoint data used in each tree; using bootstrap (i.e. selected with replacement) samples is strongly advised for even non-private trees \citep{Breiman2001,Geurts2006}. While we cannot use sampling \emph{with} replacement in our algorithm due to the privacy costs, sampling \emph{without} replacement (which is what disjoint subsets achieves) reduces over-reliance on individual records, and thus variance, in the same way \citep{Breiman2001}. Having more trees also helps average out the noise caused by the Exponential Mechanism. Of course, at some point having more votes no longer provides a benefit, and in the case of our differentially private scenario there is the added downside of having to divide up the data into more disjoint subsets -- a problem that non-private decision forests are immune from, since they can sample with replacement.

Fig.~\ref{fig:num_trees_flipped} portrays a different perspective on the same scenario as Fig.~\ref{fig:num_trees_acc}. Fig.~\ref{fig:num_trees_flipped} tells us that larger numbers of trees cause more (non-empty) leaf nodes to output a label that differs from the actual most frequent label. In other words, more most frequent label outputs are ``flipped'' to an incorrect label due to the Exponential Mechanism. This is because the disjoint subsets of data are smaller when more trees are generated, which decreases the average size of $j$ in each of the leaf nodes, which decreases the probability of the Exponential Mechanism outputting the label with the highest score. Interestingly though, the proportion of incorrect predictions being caused to preserve privacy remains around 5 to 20\%, even with 30,000 records being split among 100 trees. This means that 80\% of the predictions are as accurate as possible, given the training data. With 100 trees, each unseen record has 100 votes on what its label is, and at least 80\% of the votes (from non-empty leaf nodes) are likely to be correct. Even if the unseen record falls into a lot of empty leaf nodes, 100 votes is a large enough sample that the random votes cast by empty leaf nodes will cancel each other out in most cases. These findings are supported by the prediction accuracy results seen in Fig.~\ref{fig:num_trees_acc}. This figure demonstrates that very good prediction accuracy is possible, with many datasets achieving 85\% to 90\%  accuracy in Fig.~\ref{fig:num_trees_acc_a} when using 100 trees. Given that this is under the strict conditions required to protect privacy, these results are very promising.

For all of the above reasons, and from further testing with additional forest sizes and epsilon values which agreed with the observations seen in Fig.~\ref{fig:num_trees_acc} and Fig.~\ref{fig:num_trees_flipped}, we recommend building 100 random trees with our algorithm.

\section{Additional Experiments}\label{sec:experiments}

Aside from the experiments included throughout Section~\ref{sec:forest}, we present some other results here. In Section~\ref{subsec:size} we demonstrate that with larger datasets, much smaller privacy budgets are viable, and that larger privacy budgets cause the classifier quality to become asymptotically close to a non-private classifier. In Section~\ref{subsec:opponents} we compare our algorithm to similar algorithms, introduced in Section~\ref{subsec:DP-RDTs}.

\subsection{Scaling with Dataset Size and Privacy Budget} \label{subsec:size}

To save on computation time we perform our experiments with small datasets, with 30,000 records in our synthetic datasets and 5,456 to 30,162 records in our real-world datasets. To compensate for this, we use a relatively large privacy budget of $\epsilon=1$ for most of our experiments. Fig.~\ref{fig:scaling_size} demonstrates why we can do this: one of the advantages of differential privacy is that it scales very well, adding \emph{less} noise the more data there is \citep{Dwork2014}. In other words, a sample of 3,000,000 records with $\epsilon=0.1$ can achieve comparable results to a sample of 30,000 records when $\epsilon=1.0$, as seen in Fig.~\ref{fig:scaling_size}. What privacy budget is actually acceptable for any particular scenario in the real world depends very much on the specifics of the scenario. For example, a large public project was able to use a privacy budget of $\epsilon=8.6$ \citep{Machanavajjhala2008}.

\begin{figure}[t]
	\centering
	{\footnotesize \def\svgwidth{3.5in} 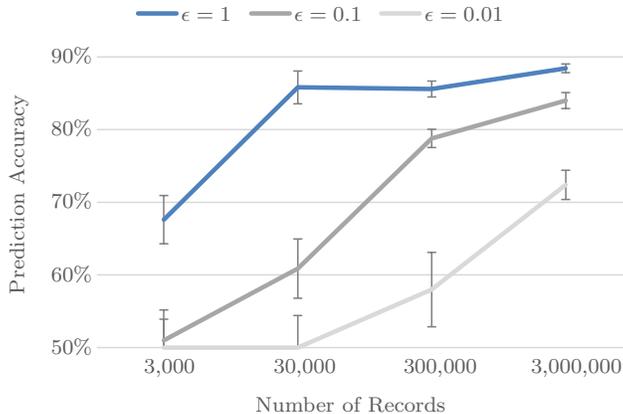}
	\caption{The average prediction accuracy of our proposed algorithm, when applied to dataset SynthF with different numbers of generated records $n$. The results of three privacy budgets are shown, $\epsilon=0.01,0.1,1.0$.}
	\label{fig:scaling_size}
\end{figure}

Differential privacy, and therefore our proposed technique, also scales well with larger privacy budgets, as seen in Fig.~\ref{fig:increasing_epsi}. Here we can see that as $\epsilon$ increases, our differentially-private technique gets asymptotically close to a non-private version of our technique (described in Section~\ref{subsec:methodology}). We include \citet{Breiman2001}'s Random Forest technique  to act as a reference point.

\begin{figure}[t]
	\centering
	{\footnotesize \def\svgwidth{4.0in} 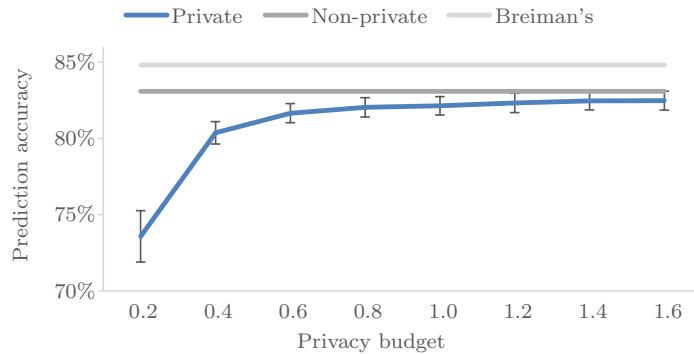}
	\caption{The average prediction accuracy of our private technique as $\epsilon$ increases for the Adult dataset, approaching the prediction accuracy of a non-private extremely random forest, with all the same parameters as our technique except that no noise is added to the most frequent labels. We also include the prediction accuracy of \citet{Breiman2001}'s Random Forest as context.}
	\label{fig:increasing_epsi}
\end{figure}

\subsection{Comparisons with Other Techniques} \label{subsec:opponents}

We implement the following differentially private decision tree algorithms, using all of their recommended parameters: \citet{Jagannathan2012}'s (henceforth called JPW); \citet{Friedman2010}'s (henceforth called FS); and \citet{Rana2016}'s (henceforth called RGV). Since RGV uses a weaker form of differential privacy, we first compare with the algorithms using the same definition as us: JPW and FS. JPW heuristically recommends a forest size of 10 trees \citep{Jagannathan2012}, while FS only builds one tree \citep{Friedman2010}. For tree depth, JPW use a depth of $d=\min(m/2,\log_bn-1)$ where $b$ is the average domain size of the features; FS uses a depth of $d=5$. We then run both of their techniques, as well as ours, on seven synthetic datasets and seven real-world datasets with $\epsilon=1$. We also run \citet{Breiman2001}'s Random Forest technique  on each of the datasets, to act as a non-private benchmark. While comparing to a technique that completely disregards privacy is obviously unfair, it provides context about the prediction accuracy that is possible under more optimal conditions for each dataset. The results are presented in Fig.~\ref{fig:vs_opponents}.

All of the results reported in this section are statistically significant. Using the Wilcoxon signed rank test (since we cannot assume that the results are normally distributed), we find that the differences between our technique and JPW, our technique and FS, and our technique and RGV are statistically significant for all datasets. The weakest significance is between our technique and FS for the GammaTele dataset, with a p-value of 0.000057.

\begin{figure}
	\centering
	{\footnotesize \def\svgwidth{4.4in} 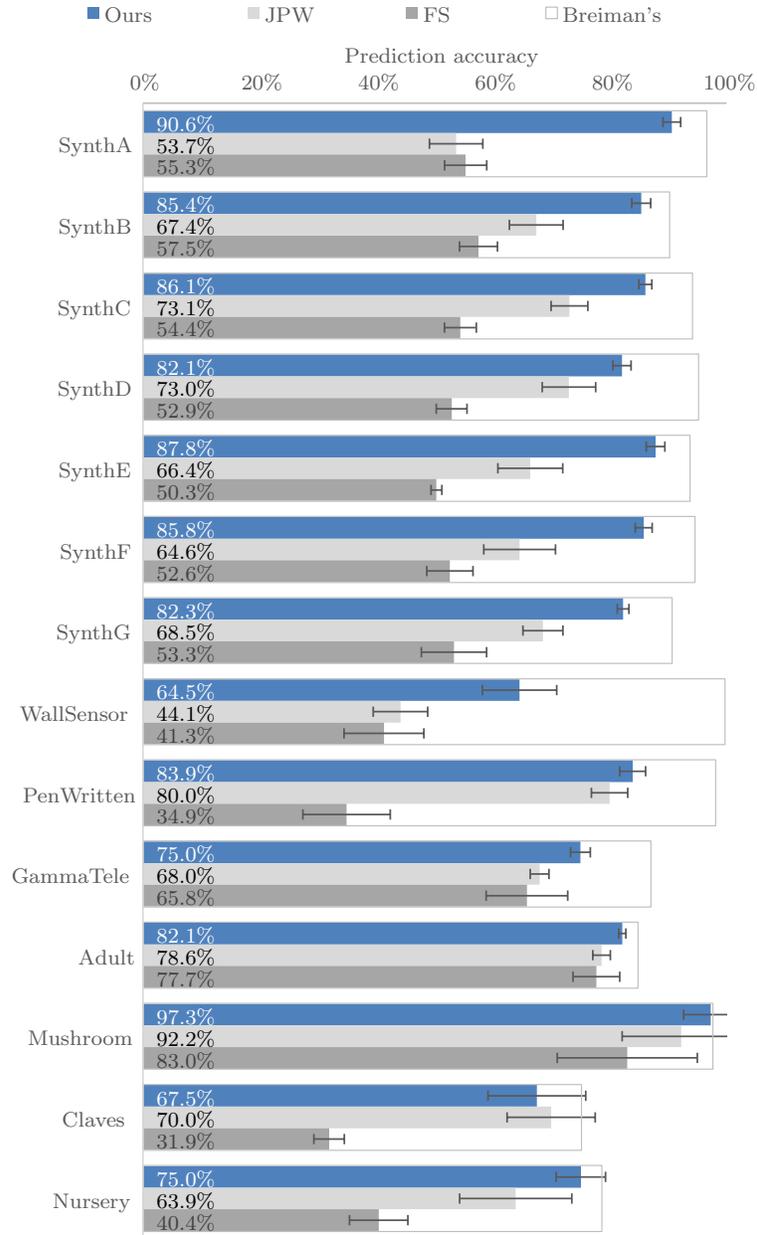}
	\caption{The average prediction accuracy of three differentially-private decision tree algorithms with $\epsilon=1$: our proposed technique, JPW \citep{Jagannathan2012}, and FS \citep{Friedman2010}. We also provide the prediction accuracy of \citet{Breiman2001}'s (non-private) Random Forest for context, portrayed as an outlined bar.}
	\label{fig:vs_opponents}
\end{figure}

The last three real-world datasets presented in Fig.~\ref{fig:vs_opponents} have discrete features only. Both of the other papers performed their experiments with only discrete features, and FS greatly prefers them due to continuous features being much more inefficient with the privacy budget \citep{Friedman2010}. Our only loss out of all datasets comes from one of these discrete datasets, Claves, where JPW achieves 2.5\% better prediction accuracy on average. 

We beat both techniques in all other cases, both when using discrete data and continuous data. While some improvements are minor (3.5\% in the case of the Adult dataset against JPW), others are very large; more than 35\% against both JPW and FS for SynthA. We often achieve more than a 10\% improvement over JPW for other datasets, and 30\% over FS.

Of course, the process of adding noise to query outputs to preserve privacy means that performing better than a technique that makes no effort to preserve privacy is all but impossible. We can see this in Fig.~\ref{fig:vs_opponents}, where \citet{Breiman2001}'s Random Forest technique  always achieves betters prediction accuracy. However, our technique performs almost as well for some synthetic datasets, as well as Adult, Mushroom and Nursery.

\begin{figure}[t]
	\centering
	{\footnotesize \def\svgwidth{4.2in} 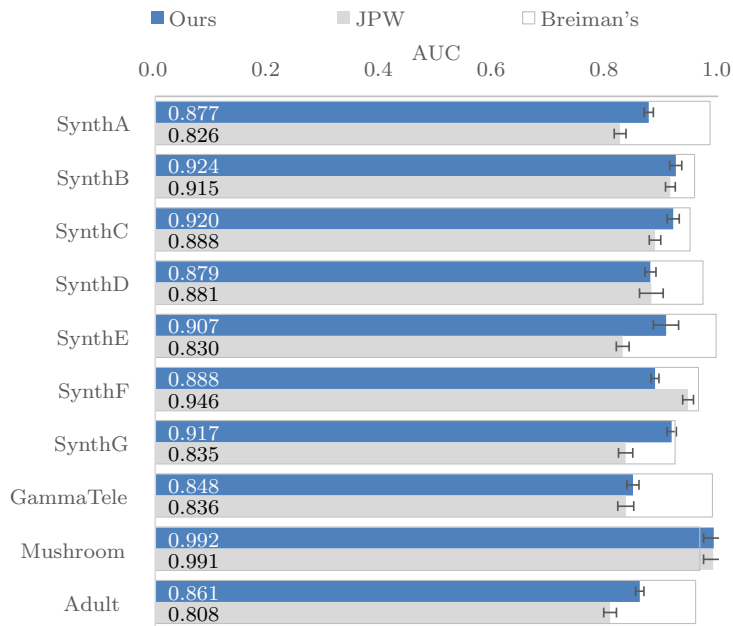}
	\caption{The average area under the ROC curve of our algorithm, compared to \citet{Jagannathan2012}'s private JPW and \citet{Breiman2001}'s non-private Random Forest, when $\epsilon=1$.}
	\label{fig:AUC}
\end{figure}

Rather than relying solely on prediction accuracy to measure the effectiveness of our algorithm, Fig.~\ref{fig:AUC} and Fig.~\ref{fig:F1} present results for AUC (Area Under the ROC Curve) \citep{Hanley1982} and F1 score (sometimes called F-measure) \citep{Rijsbergen1979} respectively. AUC and F1 score work best on datasets with binary labels, so we limit the figures to just those datasets. When calculating the true and false predictions for the positive and negative labels, we consider the least frequent class label to be the positive case. Fig.~\ref{fig:AUC} and Fig.~\ref{fig:F1} show that our algorithm out-performs JPW in 11 of the 20 results, ties in eight results (if we consider being within one standard deviation of each other a tie) and only loses in one result. This supports the findings of Fig.~\ref{fig:vs_opponents}, where JPW only beat us once, and by less than one standard deviation.

\begin{figure}[t]
	\centering
	{\footnotesize \def\svgwidth{4.2in} 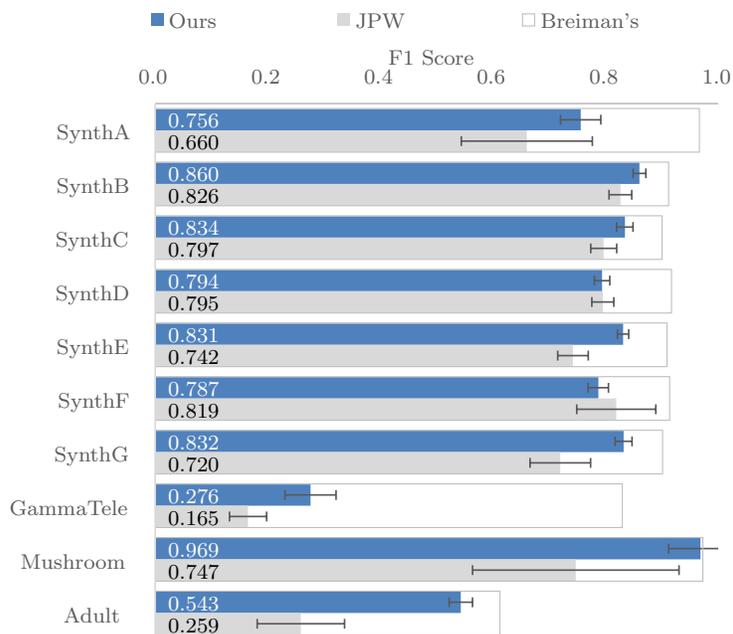}
	\caption{The average F1 score of our algorithm, compared to \citet{Jagannathan2012}'s private JPW and \citet{Breiman2001}'s non-private Random Forest, when $\epsilon=1$.}
	\label{fig:F1}
\end{figure}

RGV is a little different; they propose a weakened definition of differential privacy, and prove that a large ensemble of trees\footnote{Each tree in RGV is built similarly to FS, with some modifications.} can be made weakly differentially private with high utility. Put briefly, their differential privacy definition protects individuals from their feature values being leaked, but not from their presence in the dataset being leaked. RGV uses the entire privacy budget in each tree, and determines how many trees can be built based on a series of formulas provided in their paper \citep{Rana2016}. We run their technique with $\epsilon=1$ on the seven synthetic datasets and three real-world datasets (their paper only provides formulas for datasets with binary classes). The number of trees ranges from 110 to 1385 trees, depending on the dataset. We compare our strictly differentially private algorithm to their weakly differentially private algorithm in Fig.~\ref{fig:vs_rana}.

\begin{figure}[t]
	\centering
	{\footnotesize \def\svgwidth{3.4in} 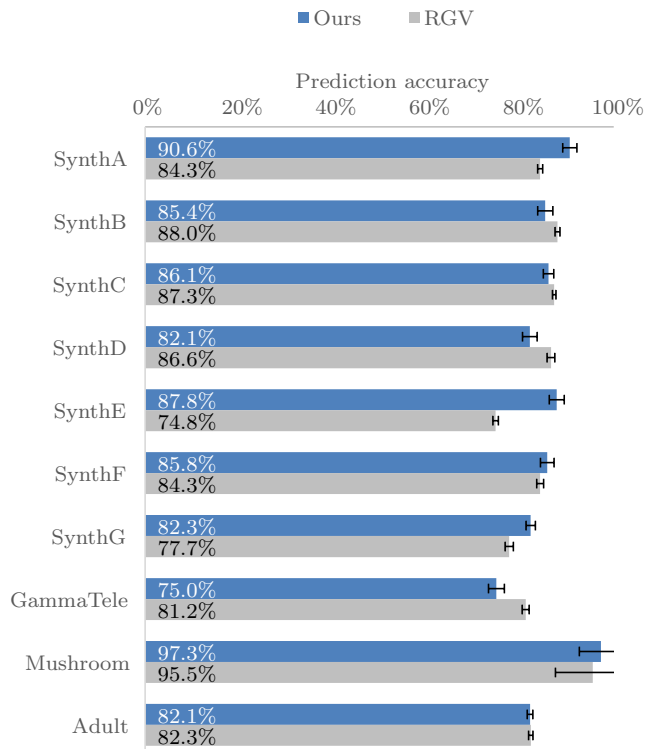}
	\caption{The average prediction accuracy of our technique compared to RGV \citep{Rana2016} when $\epsilon=1$. Note that RGV uses a weaker definition of differential privacy.}
	\label{fig:vs_rana}
\end{figure}

Even when competing against a less private algorithm, we see in Fig.~\ref{fig:vs_rana} that our algorithm still performs well. In fact we beat RGV in five cases, lose in four cases, and tie for the Adult dataset. Our accuracy improvements range from 1.5\% to 13.0\%, and RGV's accuracy improvements range from 1.2\% to 6.2\%. This result is quite impressive given the unfair comparison, and challenges the notion that strict, $(\epsilon,0)$-differential privacy \citep{Dwork2014} is too strict for high-quality data mining to be possible \citep{Rana2016,Hu2015}.

\section{Conclusion}\label{sec:discussion}

This paper proposes and explores a new method for differentially privately outputting the most (or least) frequent item in a set, using smooth sensitivity. We apply these findings to a random decision forest framework and achieve substantially higher accuracy than the state-of-the-art \citep{Friedman2010,Jagannathan2012,Fletcher2015b,Fletcher2015c,Rana2016}. We also extend the work done by \citet{Fan2003} to calculate the optimal depth for a random decision tree when using continuous features.

Each subsection in Section~\ref{sec:forest} explores a component of our algorithm, leading to several novel conclusions that improve the utility of differentially private random decision forests. Each theoretical conclusion is coupled with empirical results, demonstrating the benefits of the theory when put into practice. Finally in Section~\ref{sec:experiments}, combining all the improvements made by the previous subsections, we see that differentially private decision trees can be made from real-world data to create a highly accurate classifier, while simultaneously protecting the privacy of every person in the data. We also demonstrated in Section~\ref{subsec:opponents} that weakening the definition of differential privacy is not necessary in order to achieve good utility.

Section~\ref{subsec:size} demonstrates that with enough data, even very small privacy budgets are enough to make an accurate classifier. A user would not even have to use their entire budget on just our classifier, but can instead ask many other queries in addition to outputting our efficient classifier. On the other end of the spectrum, given enough privacy budget, a useful classifier can be built from very little data. Fig.~\ref{fig:scaling_size} demonstrates that with a budget of $\epsilon=1$, 30,000 records is all that is needed to make a classifier with over 85\% predictive accuracy.

Our findings in Section~\ref{subsec:majority} are quite generalizable; now that it has been proven that the smooth sensitivity of queries that output the most (or least) frequent item in a set is $e^{-j\epsilon}$, future research (of both ourselves and others) can explore constructing other applications with similar queries. Many machine learning domains can take advantage of queries that output the most frequent item, such as frequent pattern mining \citep{Bhaskar2010}.

We hope that our findings on smooth sensitivity, disjoint data and tree depth with continuous features aid researchers in their future work, and that our high-accuracy classifier aids data scientists in building differentially private applications.


~

\textbf{Sam Fletcher} is currently completing his Ph.D. at the School of Computing and Mathematics, Charles Sturt University, Australia. His research interests focus on decision trees, differential privacy, and the utility trade-offs required to ensure privacy. His biography is available at \emph{http://samfletcher.work}.

~

\textbf{Md Zahidul Islam} is an Associate Professor in Computer Science in the School of Computing and Mathematics, Charles Sturt University, Australia. His main research interests include data preprocessing and cleansing, various data mining algorithms, applications of data mining techniques, and privacy issues related to data mining. Find his biography at \emph{http://csusap.csu.edu.au/{\raise.17ex\hbox{$\scriptstyle\mathtt{\sim}$}}zislam/}. 

\newpage

\begin{centering}
\part*{\centering\Large Differentially-Private Random Forests using the Report One-Sided Noisy Arg-Max Mechanism}
\subsection*{\large Sam Fletcher}
\subsection*{Addendum, August 2021}
\end{centering}

\setcounter{section}{0}
\section{Introduction}
This working paper is an addendum to Fletcher and Islam's 2017 paper, ``Differentially Private Random Decision Forests using Smooth Sensitivity'' \cite{aFletcher2017}. In that paper, random decision trees are generated and the majority class labels in each leaf are reported using the Exponential mechanism \cite{aMcSherry2007}. Smooth sensitivity \cite{aNissim2007a} is applied to the Exponential mechanism in order to reduce the amount of noise required.

In August 2021, during Mete Ismayil's independent implementation of the original algorithm in IBM's Diffprivlib library\footnote{\url{https://github.com/IBM/differential-privacy-library}}, Naoise Holohan raised\footnote{\url{https://github.com/IBM/differential-privacy-library/pull/41\#pullrequestreview-731658503}} the following concern:

\begin{quote}
    I'm not sure if they have implemented smooth sensitivity as it should be. You typically can't use a smooth sensitivity bound as you would a global sensitivity bound. For example, while you can use Laplace noise with global sensitivity to achieve (pure) differential privacy, smooth sensitivity requires Cauchy noise to satisfy pure differential privacy. I'm not aware of a proof showing smooth sensitivity being used as a replacement for global sensitivity in the exponential mechanism. (In fact, it seems I can construct examples that break differential privacy when using a set-up of the exponential mechanism proposed in this paper).
\end{quote}

Upon further investigation, this criticism appears to be valid: applying smooth sensitivity to the Exponential mechanism has not yet been proven. The preliminary experiments mentioned in the quote above indicate that the true privacy cost $\epsilon$ may be twice the number reported in the original paper. While doubling the privacy cost may be acceptable sometimes, the lack of a formal proof makes the implementation questionable.

Section 2 proposes a solution to the above issue, removing the need for smooth sensitivity.

\section{Report One-Sided Noisy Arg-Max}

Report One-Sided Noisy Arg-Max (ROSNAM) is a special case of the Exponential mechanism, described by Dwork and Roth below \cite{aDwork2014}:

\begin{quote}
    \textbf{Example 3.6} (Best of Two). Consider the simple question of determining which of exactly two medical conditions $A$ and $B$ is more common. Let the two true counts be 0 for condition $A$ and $f_c > 0$ for condition $B$. Our notion of utility will be tied to the actual counts, so that conditions with bigger counts have higher utility and $\Delta u = 1$. Thus, the utility of $A$ is 0 and the utility of $B$ is $f_c$. Using the Exponential Mechanism we can immediately apply Corollary 3.12 to see that the probability of observing (wrong) outcome $A$ is at most $2e^{-f_c(\epsilon/(2\Delta u))} = 2e^{-f_c\epsilon/2}$.
   
    ...
   
    Consider the \textit{Report One-Sided Noisy Arg-Max} mechanism, which adds noise to the utility of each potential output drawn from the one-sided exponential distribution with parameter $\epsilon/\Delta u$ in the case of a monotonic utility, or parameter $\epsilon/2\Delta u$ for the case of a non-monotonic utility, and reports the resulting arg-max.
\end{quote}

This fits our scenario: we wish to report the class with the maximum count, where the addition of a record can only ever increase the utility (i.e. the counts are monotonic).

This means that rather than using the piece-wise utility function described in Theorem 1 of the original paper \cite{aFletcher2017}, the utility of a given class $c$ can be defined as $u = f_c$, where $f_c$ is the frequency of the class count.

In turn, rather than the sensitivity of the utility function being $e^{-j\epsilon}$ (where $j$ is the difference between the most frequent and second-most frequent class counts) \cite{aFletcher2017}, the sensitivity becomes $\Delta u = 1$ since the utility $u$ is a simple count.

As a simple example, if we have two classes A and B with frequencies $f_A=5$ and $f_B=10$, and a privacy cost of $\epsilon=0.1$, the resulting probability of outputting each class was previously:
\begin{equation}
    Pr(A) \propto \exp(0) = 1
\end{equation}
\begin{equation}
    Pr(B) \propto \exp(\frac{0.1}{2\times0.36}) = 1.14
\end{equation}
and with the updated equation, becomes:
\begin{equation}
    Pr(A) \propto \exp(0.1\times5) = 1.6
\end{equation}
\begin{equation}
    Pr(B) \propto \exp(0.1\times10) = 2.7 .
\end{equation}

With this change, the remainder of the original paper's algorithm is unaffected. Table~\ref{tab:changes} summarizes the necessary changes to the original paper. New empirical results will be required, but as can be seen in Table~\ref{tab:examples}, we can expect the results to improve.

\begin{table}[t]
    \centering
    \renewcommand{\arraystretch}{1.4}
    \setlength{\tabcolsep}{11pt}
    \begin{tabular}{c|c|c}
        \textbf{Parameter} & \textbf{Original Formula} & \textbf{Updated Formula} \\
        \hline
        $u$ & 1 if $f_c = \max_c f_c$ else 0 & $f_c$ \\
        $\Delta u$ & $e^{-j\epsilon}$ & 1 \\
        $\exp\left(\frac{\epsilon u}{2\Delta u}\right)$ & $\exp\left(\frac{\epsilon (1~or~0)}{2e^{-j\epsilon}}\right)$ & $\exp(\epsilon f_c)$\\
    \end{tabular}
    \caption{The necessary changes to the original paper. The third row combines the changes in the first two rows to show the new Exponential mechanism calculation.}
    \label{tab:changes}
\end{table}

\begin{table}[t]
    \centering
    \renewcommand{\arraystretch}{1.4}
    \setlength{\tabcolsep}{11pt}
    \begin{tabular}{c c|c c}
        \textbf{$f_A$} & \textbf{$f_B$} & \textbf{Original $Pr(B)$} & \textbf{New $Pr(B)$} \\
        \hline
        5 & 10 & 52.1\% & \textbf{62.2\%} \\
        105 & 110 & 52.1\% & \textbf{62.2\%} \\
        \hline
        0 & 1 & 51.4\% & \textbf{52.5\%} \\
        0 & 10 & 53.4\% & \textbf{73.1\%} \\
        \hline
        10 & 50 & 93.8\% & \textbf{98.2\%} \\
        10 & 60 & \textbf{99.9\%} & 99.3\% \\
    \end{tabular}
    \caption{Examples of the change in probability of outputting the majority class label.}
    \label{tab:examples}
\end{table}


\end{document}